\documentclass[runningheads]{llncs}

\usepackage[firstpage]{draftwatermark}
\SetWatermarkText{\hspace*{4.1in}\raisebox{4.4in}{\includegraphics[scale=0.1]{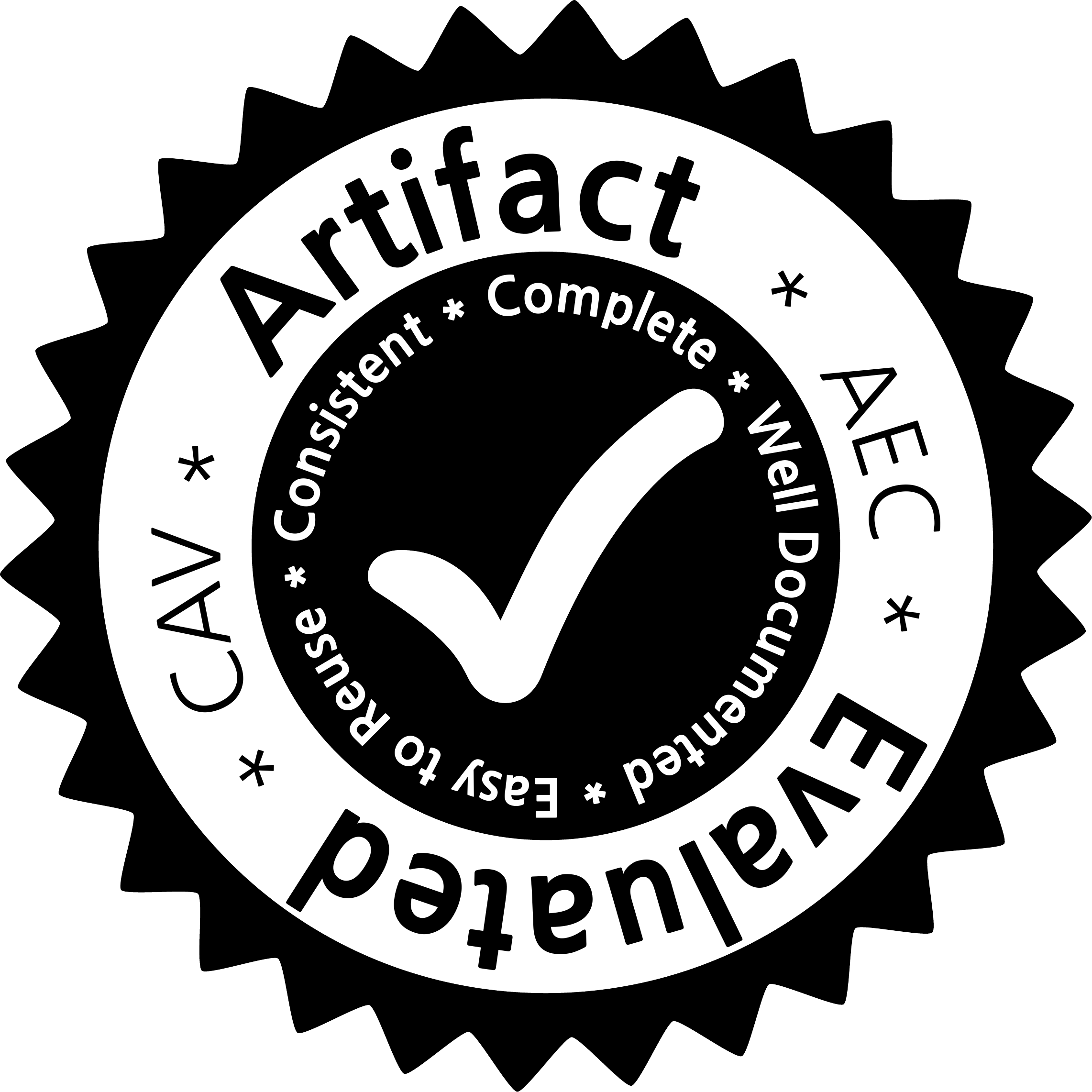}}}
\SetWatermarkAngle{0}

\usepackage{hyperref}
\usepackage{amsmath}
\usepackage{amssymb} % for \lozenge, \square
\usepackage{mathtools} % for \adjustlimits
\usepackage{enumitem} % for label= ... in enumerate env
\usepackage{subcaption} % for subfigure
\usepackage{nicefrac}
\usepackage{etoolbox} % for toggle
\usepackage{booktabs}
\usepackage{multirow}

\usepackage{ marvosym } % for \Letter symbol

\usepackage{algorithm}
\usepackage{algpseudocode}

\usepackage[disable]{todonotes} % Uncomment to hide all todo notes
\usetikzlibrary{decorations.pathreplacing}
\usetikzlibrary{decorations.pathmorphing}

\renewcommand{\mod}[1]{\widetilde{#1}} % generic way to denote a "modified version" of something
\newcommand{\emphdef}[1]{\emph{#1}} % for typesetting formal terms where they are defined
\newcommand{\refif}{``\textit{if}''} % use to refer to "directions" of an if and only if statement
\newcommand{\refonlyif}{``\textit{only if}''}
\newcommand{\lexabbr}{\mathsf{lex}} % used as general abbreviaten for "lexicographic" is math symbols

\newcommand{\aligncomment}[1]{\shortintertext{\flushright (#1)}} % Use within align environment for right-aligned text comments
\newcommand{\qee}{\hfill$\triangle$}
\newcommand{\alltargets}{\mathfrak{S}}

\newcommand{\algabsorbing}{\mathtt{SolveAbsorbing}}
\newcommand{\algsingleobj}{\mathtt{SolveSingleObj}}
\newcommand{\alggeneral}{\mathtt{SolveLex}}

\newcommand{\gfp}{\mathsf{gfp}\ }
\newcommand{\leqlex}{\leq_{\mathsf{lex}}}
\newcommand{\lesslex}{<_{\mathsf{lex}}}
\newcommand{\grlex}{>_{\mathsf{lex}}}
\newcommand{\geqlex}{\geq_{\mathsf{lex}}}
\newcommand{\prob}{\mathbb{P}}
\newcommand{\dist}{\mathcal{D}}
\newcommand{\R}{\mathbb{R}}

\newcommand{\np}{\mathsf{NP}}
\newcommand{\conp}{\mathsf{coNP}}
\newcommand{\pspace}{\mathsf{PSPACE}}

\newcommand{\nexptime}{\mathsf{NEXPTIME}}
\newcommand{\conexptime}{\mathsf{coNEXPTIME}}

\newcommand{\game}{\mathcal{G}}
\newcommand{\Act}{\mathsf{Act}}
\newcommand{\mc}{\mathcal{M}}

\newcommand{\sinks}{\mathsf{Sinks}}
\renewcommand{\path}{\pi} % finite path
\newcommand{\infpath}{\pi} % infinite path

\newcommand{\actlabels}{L}

\newcommand{\plmaxsymbol}{\square}
\newcommand{\plminsymbol}{\lozenge}
\newcommand{\plmax}{\mathsf{Max}}
\newcommand{\plmin}{\mathsf{Min}}
\newcommand{\Splmax}{S_{\plmaxsymbol}}
\newcommand{\Splmin}{S_{\plminsymbol}}
\newcommand{\maxstrat}{\sigma}
\newcommand{\minstrat}{\tau}

\newcommand{\maxstrats}{\Sigma_{\plmax}}
\newcommand{\minstrats}{\Sigma_{\plmin}}
\newcommand{\maxstratsmd}{\Sigma_{\plmax}^{\mathsf{MD}}}
\newcommand{\minstratsmd}{\Sigma_{\plmin}^{\mathsf{MD}}}
\newcommand{\bothstrats}{{\maxstrat,\minstrat}}
\newcommand{\val}{v}
\newcommand{\lexval}{\mathbf{v}^{\lexabbr}}
\newcommand{\valzeroset}{\mathsf{Zero}}

\newcommand{\reach}[1]{\mathsf{Reach}\left(#1\right)}
\newcommand{\safe}[1]{\mathsf{Safe}\left(#1\right)}
\newcommand{\until}[2]{#1\ \mathsf{U}\ #2}
\newcommand{\type}{\mathsf{type}}

\newcommand{\quanfun}{q} % function in the quantified reach/safety props
\newcommand{\indicesreach}{R} % set of indices of objectives which are reach in a \objvector

\newcommand{\objvector}{\vec{\Omega}}
\newcommand{\objvec}{\objvector}
\newcommand{\qobjvec}{\mathsf{q}\objvec}
\newcommand{\obj}{\Omega}
\newcommand{\safeop}{\mathcal{S}}
\newcommand{\reachop}{\mathcal{R}}

\usetikzlibrary{
	shapes,
	automata,
	arrows,
	arrows.meta,
	calc
}

\tikzset{
diagonal fill/.style 2 args={fill=#2, path picture={
\fill[#1, sharp corners] (path picture bounding box.south west) -|
                         (path picture bounding box.north east) -- cycle;}},
reversed diagonal fill/.style 2 args={fill=#2, path picture={
\fill[#1, sharp corners] (path picture bounding box.north west) |- 
                         (path picture bounding box.south east) -- cycle;}}
}

\tikzstyle{max}=[rectangle,draw=black,thick,inner sep=1.25mm,minimum size=5mm]
\tikzstyle{target}=[accepting] % from automata library
\tikzstyle{min}=[diamond,draw=black,thick,inner sep=0.5mm,minimum size=7mm]
\tikzstyle{prob}=[circle,draw=black,thick,inner sep=0.75mm] % don't put text in prob nodes
\tikzstyle{mcstate}=[circle,draw=black,thick,inner sep=1mm]
\tikzstyle{state}=[circle,draw=black,thick,inner sep=1mm,minimum size=8mm] % for automata
\tikzstyle{trans}=[-{Latex[scale=1.0]}]
\tikzstyle{col1}=[fill=white]
\tikzstyle{col2}=[fill=white]
\tikzstyle{col3}=[fill=white]
\tikzstyle{both cols}=[fill=white]
\tikzstyle{smaller}=[scale=0.75]
\tikzset{every loop/.style={min distance=5mm,looseness=2.5}}

\tikzstyle{pcurve}=[very thick,line cap= round]

\newtoggle{notarxiv}
\newcommand{\ifarxivelse}[2]{\iftoggle{notarxiv}{#1}{#2}}
\togglefalse{notarxiv} %for arxiv version (references to appendix, app included 
\begin{document}
\title{Stochastic Games with Lexicographic Reachability-Safety Objectives
\thanks{This research was funded in part by the TUM IGSSE Grant 10.06 (PARSEC), the German Research Foundation (DFG) project KR 4890/2-1 “Statistical Unbounded Verification”, the ERC CoG 863818 (ForM-SMArt), the Vienna Science and Technology Fund (WWTF) Project ICT15- 003, and the RTG 2236 UnRAVeL.}}
%
%\titlerunning{Abbreviated paper title}
% If the paper title is too long for the running head, you can set
% an abbreviated paper title here
%
\author{
Krishnendu Chatterjee\inst{1}\ifarxivelse{\orcidID{0000-0002-4561-241X}}{} \and
Joost-Pieter Katoen\inst{3}\ifarxivelse{\orcidID{0000-0002-6143-1926}}{} \and
Maximilian Weininger\inst{2}\ifarxivelse{\orcidID{0000-0002-0163-2152}}{} \and
Tobias Winkler\textsuperscript{(\Letter),} \inst{3}\ifarxivelse{\orcidID{0000-0003-1084-6408}}{}
}

\authorrunning{K. Chatterjee et al.}
% First names are abbreviated in the running head.
% If there are more than two authors, 'et al.' is used.
%

\institute{
IST Austria, Klosterneuburg, Austria \and
Technical University of Munich, Munich, Germany \and
RWTH Aachen University, Aachen, Germany \\
\email{tobias.winkler@cs.rwth-aachen.de}
}

\maketitle              % typeset the header of the contribution
\begin{abstract}
We study turn-based stochastic zero-sum games with lexicographic preferences over reachability and safety objectives. 
Stochastic games are standard models in control, verification, and synthesis of stochastic reactive systems that exhibit both randomness as well as angelic and demonic non-determinism.  
Lexicographic order allows to consider multiple objectives with a strict preference order over the satisfaction of the objectives.
To the best of our knowledge, stochastic games with lexicographic objectives have not been studied before. 
We establish determinacy of such games and present strategy and computational complexity results.
For strategy complexity, we show that lexicographically optimal strategies exist that are deterministic and 
memory is only required to remember the already satisfied and violated objectives.
For a constant number of objectives, we show that the relevant decision problem is in $\np \cap \conp$, matching the current known bound for single objectives; and 
in general the decision problem is $\pspace$-hard and can be solved in $\nexptime \cap \conexptime$.
We present an algorithm that computes the lexicographically optimal strategies via a reduction to computation of optimal strategies in a sequence of single-objectives games.
We have implemented our algorithm and report experimental results on various case studies.
%\keywords{First keyword  \and Second keyword \and Another keyword.}
\end{abstract}

\section{Introduction}

{\em Simple stochastic games (SGs)}~\cite{Con92} are zero-sum turn-based stochastic games played over 
a finite state space by two adversarial players, the Maximizer and Minimizer, 
along with randomness in the transition function.
These games allow the interaction of angelic and demonic non-determinism as well as stochastic 
uncertainty. 
They generalize classical models such as Markov decision processes (MDPs)~\cite{Puterman} 
which have only one player and stochastic uncertainty. 
An objective specifies a desired set of trajectories of the game, and the goal of the Maximizer is 
to maximize the probability of satisfying the objective against all choices of the Minimizer.
The basic decision problem is to determine whether the Maximizer can ensure satisfaction of the objective
with a given probability threshold.
This problem is among the rare and intriguing combinatorial problems that are $\np \cap \conp$,  
and whether it belongs to P is a major and long-standing open problem.
Besides the theoretical interest, SGs are a standard model in control and verification of stochastic 
reactive systems~\cite{Puterman,BK08,FV97,CH12},
as well as they provide robust versions of MDPs when precise transition probabilities are not known~\cite{DBLP:conf/fossacs/ChatterjeeSH08,WMK19}.

The multi-objective optimization problem is relevant in the analysis of systems with multiple, potentially conflicting 
goals, and a trade-off must be considered for the objectives. 
While the multi-objective optimization has been extensively studied for MDPs with various classes of objectives~\cite{Puterman,Altman,tacas20}, 
the problem is notoriously hard for SGs. 
Even for multiple reachability objectives, such games are 
not determined~\cite{DBLP:conf/mfcs/ChenFKSW13} and their decidability is still open.

This work considers SGs with multiple reachability and safety objectives with lexicographic preference order 
over the objectives. 
That is, we consider SGs with several objectives where each objective is either reachability or safety, and
there is a total preference order over the objectives. 
The motivation to study such lexicographic objectives is twofold.
First, they provide an important special case of general multiple objectives.
Second, lexicographic objectives are useful in many scenarios. 
For example, (i)~an autonomus vehicle might have a primary objective to avoid clashes and a secondary objective
to optimize performance; and (b)~a robot saving lives during fire in a building might have a primary objective
to save as many lives as possible, and a secondary objective to minimize energy consumption.
Thus studying reactive systems with lexicographic objectives is a very relevant problem which has been considered
in many different contexts~\cite{fishburn1974exceptional,blume1991lexicographic}.
In particular non-stochastic games with lexicographic objectives~\cite{BCHJ09,CJLS17} and MDPs with 
lexicographic objectives~\cite{WZM15} have been considered, but to the best of our knowledge
SGs with lexicographic objectives have not been studied.

In this work we present several contributions for SGs with lexicographic reachability and safety objectives.
The main contributions are as follows.
\begin{itemize}
	\item {\em Determinacy.} In contrast to SGs with multiple objectives that are not determined, we establish 
	determinacy of SGs with lexicographic combination of reachability and safety objectives. 
	
	\item {\em Computational complexity.} 
	For the associated decision problem we establish the following:
	(a)~if the number of objectives is constant, then the decision problem lies in $\np \cap \conp$, matching the
	current known bound for SGs with a single objective;
	(b)~in general the decision problem is $\pspace$-hard and can be solved in $\nexptime \cap \conexptime$.
	
	\item {\em Strategy complexity.} 
	We show that lexicographically optimal strategies exist that are deterministic but require finite memory.
	We also show that memory is only needed in order to remember the already satisfied and violated objectives.
	
	\item {\em Algorithm.} 
	We present an algorithm that computes the unique lexicographic value and the witness lexicographically optimal strategies 
	via a reduction to computation of optimal strategies in a sequence of single-objectives games. 
	
	\item {\em Experimental results.} We have implemented the algorithm and present experimental results on several case studies.
	
	%\item {\bf KRISH: LET'S OMIT THIS PART FROM CONFERENCE SUBMISSION: MENTIONS A DRAWBACK WHICH NON-EXPERT WILL JUMP TO.} 
	%We introduce a \emph{coefficient of robustness} (name TBD) to compensate for the main drawback of lexicographic preferences, their lack of robustness.
\end{itemize}

\smallskip\noindent{\em Technical contribution.}
The key idea is that, given the lexicographic order of the objectives, we can consider them sequentially.
After every objective, we remove all actions that are not optimal, thereby forcing all following computation to consider only locally optimal actions.
The main complication is that local optimality of actions does not imply global optimality when interleaving reachability and safety, as the latter objective can use locally optimal actions to stay in the safe region without reaching the more important target.
We introduce quantified reachability objectives as a means to solve this problem.

\paragraph{Related work} We present related works on: 
(a)~MDPs with multiple objectives; (b)~SGs with multiple objectives; (c)~lexicographic objectives 
in related models; and (d)~existing tool support. 

(a) MDPs with multiple objectives have been widely studied over a long time~\cite{Puterman,Altman}.
In the context of verifying MDPs with multiple objectives, both qualitative objectives such as reachability and LTL~\cite{DBLP:journals/lmcs/EtessamiKVY08}, 
as well as quantitative objectives, such as mean payoff~\cite{DBLP:conf/fsttcs/Chatterjee07a,DBLP:journals/corr/abs-1104-3489}, discounted sum~\cite{DBLP:conf/lpar/ChatterjeeFW13}, or total reward~\cite{DBLP:conf/tacas/ForejtKNPQ11} have been considered.
Besides multiple objectives with expectation criterion, other criteria have also been considered,
such as, combination with variance~\cite{DBLP:conf/lics/BrazdilCFK13}, or multiple percentile (threshold) queries~\cite{FKR95,DBLP:journals/corr/abs-1104-3489,DBLP:journals/fmsd/RandourRS17,DBLP:journals/lmcs/ChatterjeeKK17}.
Practical applications of MDPs with multiple objectives are described in~\cite{DBLP:conf/csl/BaierDK14,DBLP:conf/fase/BaierDKDKMW14,roijers2017multi}.

(b) More recently, SGs with multiple objectives have been considered, but the results are more limited \cite{DBLP:journals/ejcon/SvorenovaK16}. 
Multiple mean-payoff objectives were first examined in~\cite{DBLP:conf/tacas/BassetKTW15} and the qualitative problems are coNP-complete~\cite{DBLP:conf/lics/Chatterjee016}.
Some special classes of SGs (namely stopping SGs) have been solved for total-reward objectives~\cite{DBLP:conf/mfcs/ChenFKSW13} 
and applied to autonomous driving \cite{DBLP:conf/qest/ChenKSW13}.
However, even for the most basic question of solving SGs with multiple reachability objectives, decidability remains open. 

(c) The study of lexicographic objectives has been considered in many different contexts~\cite{fishburn1974exceptional,blume1991lexicographic}.
Non-stochastic games with lexicographic mean-payoff objectives and parity conditions have been studied in~\cite{BCHJ09} for the synthesis of 
reactive systems with performance guarantees.
Non-stochastic games with multiple $\omega$-regular objectives equipped with a monotonic preorder, which subsumes lexicographic order, have been
studied in~\cite{BHR18a}.
Moreover, the beyond worst-case analysis problems studied in~\cite{DBLP:journals/iandc/BruyereFRR17} also considers primary and secondary 
objectives, which has a lexicographic flavor. 
MDPs with lexicographic discounted-sum objectives have been studied in~\cite{WZM15}, and have been extended with partial-observability in~\cite{WZ15}.
However, SGs with lexicographic reachability and safety objectives have not been considered so far.

(d) PRISM-Games \cite{DBLP:journals/sttt/KwiatkowskaPW18} provides tool support for several multi-player multi- objective settings. 
%%the single-dimensi\-onal-focused tools GAVS+ \cite{DBLP:conf/tacas/ChengKLB11} and GIST \cite{DBLP:conf/cav/ChatterjeeHJR10} are not maintained any more.
MultiGain \cite{DBLP:conf/tacas/BrazdilCFK15} is limited to generalized mean-payoff MDPs. 
\textsc{Storm} \cite{DBLP:conf/cav/DehnertJK017} can, among numerous single-objective problems, solve Markov automata with multiple timed reachability or expected cost objectives \cite{DBLP:conf/cav/QuatmannJK17}, multi-cost bounded reachability MDPs \cite{DBLP:conf/tacas/HartmannsJKQ18}, and it can provide simple strategies for multiple expected reward objectives in MDPs~\cite{tacas20}.

\paragraph{Structure of this paper.}
After recalling preliminaries and defining the problem in Section \ref{sec:preliminaries}, 
we first consider games where all target sets are absorbing in Section \ref{sec:absorbing}. 
Then, in Section \ref{sec:general} we extend our insights to general games, yielding the full algorithm and the theoretical results. 
%Section \ref{sec:robust} discusses the coefficient of robustness.
Finally, Section \ref{sec:exp} describes the implementation and experimental evaluation. 
Section \ref{sec:conc} concludes.

\section{Preliminaries}
\label{sec:preliminaries}

\subsubsection{Notation.}
A probability distribution on a finite set $A$ is a function $f: A \to [0,1]$ such that $\sum_{x\in A}f(x) =1$. We denote the set of all probability distributions on $A$ by $\dist(A)$. 
Vector-like objects $\vec{x}$ are denoted in a bold font and we use the notation $\vec{x}_i$ for the $i$-th component of $\vec{x}$. 
We use $\vec{x}_{<n}$ as a shorthand for $(\vec{x}_1,\ldots,\vec{x}_{n-1}).$

\subsection{Basic Definitions}                                        

\subsubsection{Probabilistic Models.}
In this paper, we consider \emphdef{(simple) stochastic games}~\cite{Con92}, which are defined as follows. Let $\actlabels = \{a,b,\ldots\}$ be a finite set of actions labels.
\begin{definition}[SG]
	A \emphdef{stochastic game} (SG) is a tuple $\game = (\Splmax,\Splmin,\Act, P)$ with $S := \Splmax \uplus \Splmin \neq \emptyset$ a finite set of states, $\Act: S \rightarrow 2^\actlabels \setminus \{\emptyset\}$ defines finitely many actions available at every state, and $P: S \times \actlabels \rightarrow \dist(S)$ is the transition probability function. $P(s,a)$ is undefined if $a \notin \Act(s)$.
\end{definition}
We abbreviate $P(s,a)(s')$ to $P(s,a,s')$. We refer to the two players of the game as $\plmax$ and $\plmin$ and the sets $\Splmax$ and $\Splmin$ are the $\plmax$- and $\plmin$-states, respectively. As the game is \emph{turn based}, these sets partition the state space $S$ such that in each state it is either $\plmax$'s or $\plmin$'s turn.
The intuitive semantics of an SG is as follows: In every turn, the corresponding player picks one of the finitely many available actions $a \in \Act(s)$ in the current state $s$. The game then transitions to the next state according to the probability distribution $P(s,a)$. The winning conditions are not part of the game itself and need to be further specified.
\subsubsection{Sinks, Markov Decision Processes and Markov Chains.}
A state $s \in S$ is called \emphdef{absorbing} (or sink) if $P(s,a,s) = 1$ for all $a \in \Act(s)$ and $\sinks(\game)$ denotes the set of all absorbing states of SG $\game$.
A \emphdef{Markov Decision Process} (MDP) is an SG where either $\Splmin = \emptyset$ or $\Splmax = \emptyset$, i.e. a one-player game.
A \emphdef{Markov Chain} (MC) is an SG where $|\Act(s)|=1$ for all $s \in S$. 
For technical reasons, we allow countably infinite state spaces $S$ for both MDPs and MCs.

\subsubsection{Strategies.}
We define the formal semantics of games by means of \emphdef{paths} and \emphdef{strategies}.
An \emph{infinite path} $\path$ is an infinite sequence $\path = s_0 a_0 s_1 a_1 \dots \in (S \times \actlabels)^\omega$, such that for every $i \in \mathbb{N}$, $a_i\in \Act(s_i)$ and $s_{i+1} \in \{s' \mid P(s_i,a_i,s')>0\}$.
\emph{Finite path}s are defined analogously as elements of $(S \times \actlabels)^\ast \times S$.
Note that when considering MCs, every state just has a single action, so an infinite path can be identified with an element of $S^\omega$.

A strategy of player $\plmax$ is a function $\maxstrat \colon (S \times \actlabels)^* \times \Splmax \rightarrow \dist(\actlabels)$ where $\maxstrat(\path s)(s')>0$ only if $s \in \Act(s)$. 
It is \emphdef{memoryless} if $\maxstrat(\path s) = \maxstrat(\path' s)$ for all $\path, \path' \in (S \times \actlabels)^*$.
More generally, $\maxstrat$ has memory of class-size at most $m$ if the set $(S \times \actlabels)^*$ can be partitioned in $m$ classes $M_1,\ldots,M_m \subseteq (S \times \actlabels)^*$ such that $\maxstrat(\path s) = \maxstrat(\path' s)$ for all $1 \leq i \leq m$, $\path, \path' \in M_i$ and $s \in \Splmax$. A memory of class-size $m$ can be represented with $\lceil \log(m) \rceil$ bits.

A strategy is \emphdef{deterministic} if $\maxstrat(\path s)$ is Dirac for all $\path s$. Strategies that are both memoryless and deterministic are called \emphdef{MD} and can be identified as functions $\maxstrat \colon \Splmax \rightarrow \actlabels$. Notice that there are at most $|\actlabels|^{\Splmax}$ different MD strategies, that is, exponentially many in $\Splmax$; in general, there can be uncountably many strategies.

Strategies $\minstrat$ of player $\plmin$ are defined analogously, with $\Splmax$ replaced by $\Splmin$. The set of all strategies of player $\plmax$ is denoted with $\maxstrats$, the set of all MD strategies with $\maxstratsmd$, and similarly $\minstrats$ and $\minstratsmd$ for player $\plmin$. 

Fixing a strategy $\maxstrat$ of one player in a game $\game$ yields the \emphdef{induced MDP} $\game^\maxstrat$. Fixing a strategy $\minstrat$ of the second player too, yields the \emphdef{induced MC} $\game^\bothstrats$. Notice that the induced models are finite if and only if the respective strategies use finite memory.

Given an (induced) MC $\game^{\bothstrats}$, we let $\prob_s^\bothstrats$ be its associated probability measure on the Borel-measurable sets of infinite paths obtained from the standard cylinder construction where $s$ is the initial state~\cite{Puterman}.

\subsubsection{Reachability and Safety.}
%In our setting, an objective for a game $\game=(S,\Act,P)$ is a Borel-measurable set $\obj \subseteq S^\omega$ of infinite paths in the game graph.
%Fix an MC with state space $S$.
In our setting, a \emph{property} is a Borel-measurable set $\obj \subseteq S^\omega$ of infinite paths in an SG.
The \emphdef{reachability property} $\reach{T}$ where $T \subseteq S$ is the set $\reach{T} = \{s_0s_1\ldots \in S^\omega \mid \exists i \geq 0 \colon s_i \in T\}$. The set $\safe{T} = S^\omega \setminus \reach{T}$ is called a \emphdef{safety property}. 
Further, for sets $T_1,T_2 \subseteq S$ we define the \emphdef{until property} $\until{T_1}{T_2} = \{s_0s_1\ldots \in S^\omega \mid \exists i \geq 0 \colon s_i \in T_2 \wedge \forall j < i \colon s_j \in T_1\}$. 
These properties are measurable (e.g. \cite{BK08}). A reachability or safety property where the set $T$ satisfies $T \subseteq \sinks(\game)$ is called \emphdef{absorbing}. 
For the safety probabilities in an (induced) MC, it holds that $\prob_s(\safe{T})= 1 - \prob_s(\reach{T})$. 
We highlight that an objective $\safe{T}$ is specified by the set of paths to avoid, i.e. paths satisfying the objective remain forever in $S \setminus T$. 
%This notation differs from the one used by other authors but proved convenient in our case.\todo{M: I don't like this formulation; also Jan said, that both are kinda standard, so we just stress which one we are using.}

%Given a game $\game = (S,\Act,P)$ and strategies $\bothstrats$ of both players, we define the \emphdef{induced Markov Chain} $\game^\bothstrats$ with (countably infinite) state space $S^\bothstrats = (S \times \Act)^* \times S$ and the transition probability function $P^\bothstrats$ which is obtained by following $\maxstrat$ in states $\path s$ with $s \in \Splmax$ and $\minstrat$ otherwise. Notice that if $\bothstrats$ are memoryless, then $\game^\bothstrats$ can be defined such that $S^\bothstrats = S$. The probability measure associated with the induced MC is denoted with $\prob_s^\bothstrats$, where $s \in S^\bothstrats$.

\subsection{Stochastic Lexicographic Reachability-Safety Games}

%Note that to the best of our knowledge, this is the first work on stochastic lexicographic reachability-safety games; however, the model is very natural and a straightforward extension of the ideas of e.g.~\cite{WZ15}.

SGs with lexicographic preferences are a straightforward adaptation of the ideas of e.g.~\cite{WZ15} to the game setting.
The \emphdef{lexicographic} order on $\R^n$ is defined as $\vec{x} \leqlex \vec{y}$ iff $\vec{x}_i \leq \vec{y}_i$ where $i \leq n$ is the greatest position such that for all $j < i$ it holds that $\vec{x}_j = \vec{y}_j$. The position $i$ thus acts like a \emphdef{tiebreaker}. Notice that for arbitrary sets $X \subseteq [0,1]^n$, suprema and infima exist in the lexicographic order.

%\subsubsection{Lexicographic Objectives, Optimal Strategies, Values}
\begin{definition}[Lex-Objective and Lex-Value]
	\label{def:lex_obj}
	A \emphdef{lexicographic reachability- safety objective} (\emphdef{lex-objective}, for short) is a vector $\objvector = (\obj_1,\ldots,\obj_n)$ such that $\obj_i \in \{\reach{S_i}, \safe{S_i}\}$ with $S_i \subseteq S$ for all $1\leq i \leq n$. 
	We call $\objvec$ \emphdef{absorbing} if all the $\obj_i$ are absorbing, i.e., if $S_i \subseteq \sinks(\game)$ for all $1 \leq i \leq n$.
	The \emphdef{lex-(icographic)value} of $\objvector$ at state $s \in S$ is defined as: 
	\begin{equation}
	\label{eq:def_lexval}
	^{\objvec}\lexval(s) = \adjustlimits \sup_{\maxstrat \in \maxstrats} \inf_{\minstrat \in \minstrats} \prob^\bothstrats_s(\objvector)
	\end{equation}
	where $\prob^\bothstrats_s(\objvector)$ denotes the \emph{vector} $(\prob^\bothstrats_s(\obj_1),\ldots,\prob^\bothstrats_s(\obj_n))$ and the suprema and infima are taken with respect to the order $\leqlex$ on $[0,1]^n$.
\end{definition}

Thus the lex-value at state $s$ is the lexicographically supremal vector of probabilities that $\plmax$ can ensure against all possible behaviors of $\plmin$. 
We will prove in Section \ref{sec:theorImpl} that the supremum and infimum in \eqref{eq:def_lexval} can be exchanged; this property is called \emphdef{determinacy}. 
We omit the superscript $\objvec$ in $^{\objvec}\lexval$ if it is clear from the context. We also omit the sets $\maxstrats$ and $\minstrats$ in the suprema in \eqref{eq:def_lexval}, e.g. we will just write $\sup_{\maxstrat}$.

\begin{figure}[t]
	\centering
	\begin{subfigure}{0.45 \textwidth}
	\begin{tikzpicture}[scale=1.15]
	\node[min] (p) at (0,3) {$p$};
	\node[max] (q) at (0,1.5) {$q$};
	\node[max] (r) at (0,0) {$r$};
	\node[max] (s) at (2,3) {$s$};
	\node[min] (t) at (2,2) {$t$};
	\node[min] (u) at (2,1) {$u$};
	\node[min] (v) at (2,0) {$v$};
	\node[max] (w) at (4,0) {$w$};
	\node[prob] (p1) at (0.7,1.5) {};
	\node[prob] (p2) at (0.7,0.5) {};
	\node[prob] (p3) at (3,0) {};
	
	\draw[trans] (p) edge (s);
	\draw[trans] (p) edge (q);
	\draw[trans] (q) edge[bend left=20] (r);
	\draw[trans] (r) edge[bend left=20] (q);
	\draw[trans] (r) edge[-] (p1);
	\draw[trans] (r) edge[-] (p2);
	\draw[trans] (v) edge[-] (p3);
	\draw[trans] (p1) edge (t);
	\draw[trans] (p1) edge (u);
	\draw[trans] (p2) edge (t);
	\draw[trans] (p2) edge (v);
	\draw[trans] (p3) |- (u);
	\draw[trans] (p3) edge (w);
	
%	\draw[trans] (s) edge[loop above] (s);
%	\draw[trans] (t) edge[loop above] (t);
%	\draw[trans] (u) edge[loop above] (u);
%	\draw[trans] (w) edge[loop above] (w);

	\draw[green!50!black,thick] ($(s.north west)+(-0.35,0.4)$)  rectangle ($(t.south east)+(0.4,-0.4)$);
	\draw[red,thick,dotted]     ($(t.north west)+(-0.3,0.3)$) rectangle ($(u.south east)+(0.3,-0.3)$);
	\node (s1) at (2.8,3.3) {\textcolor{green!50!black}{$S_1$}};
	\node (s2) at (2.8,1.7) {\textcolor{red}{$S_2$}};
	\end{tikzpicture}
	\subcaption{}
	\label{fig:prelimRunningA}
	\end{subfigure}
	\begin{subfigure}{0.45 \textwidth}
	\begin{tikzpicture}[scale=1.15]
	\node[min] (p) at (0,3) {$p$};
	\node[max] (q) at (0,1.5) {$q$};
	\node[max] (r) at (0,0) {$r$};
	\node[max] (s) at (2,3) {$s$};
	\node[min] (t) at (2,2) {$t$};
	\node[min] (u) at (2,1) {$u$};
	\node[min] (v) at (2,0) {$v$};
	\node[max] (w) at (4,0) {$w$};
	%\node[prob] (p1) at (0.7,1.5) {};
	\node[prob] (p2) at (0.7,0.5) {};
	\node[prob] (p3) at (3,0) {};
	
	%\draw[trans] (p) edge (s);
	\draw[trans] (p) edge (q);
	\draw[trans] (q) edge[bend left=20] (r);
	\draw[trans] (r) edge[bend left=20] (q);
	%\draw[trans] (r) edge[-] (p1);
	\draw[trans] (r) edge[-] (p2);
	\draw[trans] (v) edge[-] (p3);
	%\draw[trans] (p1) edge (t);
	%\draw[trans] (p1) edge (u);
	\draw[trans] (p2) edge (t);
	\draw[trans] (p2) edge (v);
	\draw[trans] (p3) |- (u);
	\draw[trans] (p3) edge (w);
	
	%	\draw[trans] (s) edge[loop above] (s);
	%	\draw[trans] (t) edge[loop above] (t);
	%	\draw[trans] (u) edge[loop above] (u);
	%	\draw[trans] (w) edge[loop above] (w);
	
	\draw[green!50!black,thick] ($(s.north west)+(-0.4,0.4)$)  rectangle ($(t.south east)+(0.4,-0.4)$);
	\draw[red,thick,dotted]     ($(t.north west)+(-0.3,0.3)$) rectangle ($(u.south east)+(0.3,-0.3)$);
	\node (s1) at (2.8,3.3) {\textcolor{green!50!black}{$S_1$}};
	\node (s2) at (2.8,1.7) {\textcolor{red}{$S_2$}};
	\end{tikzpicture}
	\subcaption{}
	\label{fig:prelimRunningB}
	\end{subfigure}
	\caption{
	(a) An example of a stochastic game. 
	$\plmax$-states are rendered as squares $\plmaxsymbol$ and $\plmin$-states as rhombs $\plminsymbol$.
	Probabilistic choices are indicated with small circles. 
	In this example, all probabilities equal $\nicefrac{1}{2}$.
	The absorbing lex-objective $\objvec = \{\reach{S_1}, \safe{S_2}\}$ is indicated by the thick green line around $S_1 = \{s,t\}$ and the dotted red line around $S_2 = \{t,u\}$. Self-loops in sinks are omitted.
    (b) Restriction of the game to lex-optimal actions only.}
	\label{fig:prelimRunning}
\end{figure}
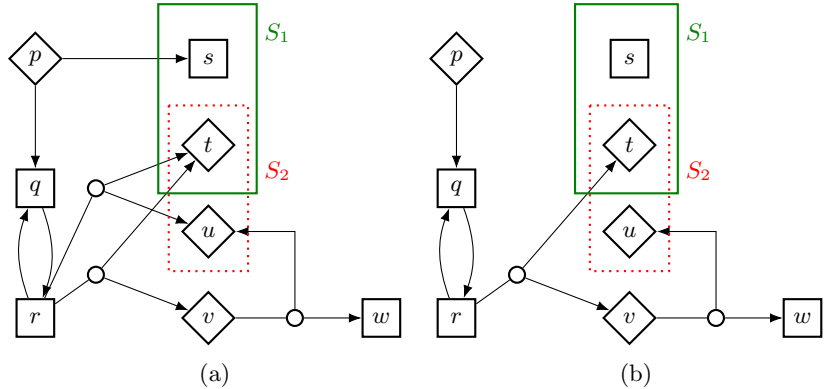

\begin{example}[SGs and lex-values]
	\label{ex:running_ex_intro}
	Consider the SG sketched in Figure \ref{fig:prelimRunningA} with the lex-objective $\objvec = \{\reach{S_1}, \safe{S_2}\}$. 
	Player $\plmax$ must thus maximize the probability to reach $S_1$ and, moreover, among all possible strategies that do so, it must choose one that maximizes the probability to avoid $S_2$ forever.
	\qee
\end{example}

\subsubsection{Lex-value of actions and lex-optimal actions.} We extend the notion of value to actions. Let $s \in S$ be a state. The \emphdef{lex-value of an action} $a \in \Act(s)$ is defined as $\lexval(s,a) = \sum_{s'}P(s,a,s')\lexval(s')$. If $s \in \Splmax$, then action $a$ is called \emphdef{lex-optimal} if $\lexval(s,a) = \max_{b \in \Act(s)}\lexval(s,b)$. Lex-optimal actions are defined analogously for states $s \in \Splmin$ by considering the minimum instead of the maximum. Notice that there is always at least one optimal action because $\Act(s)$ is finite by definition.

\begin{example}[Lex-value of actions]
	\label{ex:running_ex_intro2}
	We now intuitively explain the lex-values of all states in Figure \ref{fig:prelimRunningA}.
	The lex-value of sink states $s$, $t$, $u$ and $w$ is determined by their membership in the sets $S_1$ and $S_2$.
	E.g., $\lexval(s) = (1,1)$, as it is part of the set $S_1$ that should be reached and not part of the set $S_2$ that should be avoided.
	Similarly we get the lex-values of $t$, $u$ and $w$ as $(1,0)$, $(0,0)$ and $(0,1)$ respectively.
	State $v$ has a single action that yields $(0,0)$ or $(0,1)$ each with probability $\nicefrac 1 2$, thus $\lexval(v) = (0,\nicefrac 1 2)$.
	
	State $p$ has one action going to $s$, which would yield $(1,1)$. However, as $p$ is a $\plmin$-state, its best strategy is to avoid giving such a high value. Thus, it uses the action going downwards and $\lexval(p)=\lexval(q)$. 
	State $q$ only has a single action going to $r$, so $\lexval(q)=\lexval(r)$.
	
	State $r$ has three choices: (i)~Going back to $q$, which results in an infinite loop between $q$ and $r$, and thus never reaches $S_1$. So a strategy that commits to this action will not achieve the optimal value.
	(ii)~Going to $t$ or $u$ each with probability $\nicefrac 1 2$. In this case, the safety objective is definitely violated, but the reachability objective achieved with $\nicefrac 1 2$.
	(iii)~Going to $t$ or $v$ each with probability $\nicefrac 1 2$. Similarly to (ii), the probability to reach $S_1$ is $\nicefrac 1 2$, but additionally, there is a $\nicefrac 1 2 \cdot \nicefrac 1 2$ chance to avoid $S_2$. 
	Thus, since $r$ is a $\plmax$-state, its lex-optimal choice is the action leading to $t$ or $v$ and we get $\lexval(r) = (\nicefrac 1 2, \nicefrac 1 4)$.	
	\qee
	
	%Figure~\ref{fig:prelimRunningB} shows the game restricted to the lex-optimal actions. Note that the action from $r$ to $q$ is still present, as it can be played an arbitrary finite number of times without changing the value. 
	%The action is not suboptimal, but the strategy that commits to infinitely playing it.\todo{M: I added this last paragraph, but I think we should drop it, as this is the point of Example \ref{ex:safeEvil}. Then Fig \ref{fig:prelimRunningB} is not referenced in the prelims, but I think that is ok. T: agree}
\end{example}

%An action which is not lex-optimal is called \emphdef{lex-sub-optimal}. M: I removed all ocurrences of the word, by rephrasing; I learnt in a style course it is better to phrase things positively. (Drop lex-sub-optimal -> Keep only lex-optimal)

Notice that with the kind of objectives considered, we can easily swap the roles of $\plmax$ and $\plmin$ by exchanging safety objectives with reachability and vice versa. It is thus no loss of generality to consider subsequently introduced notions such as optimal strategies only from the perspective of $\plmax$.

\begin{definition}[Lex-Optimal Strategies]
	A strategy $\maxstrat \in \maxstrats$ is \emphdef{lex-optimal} for $\objvector$ if for all $s \in S$, 
		$\lexval(s) = \inf_{\minstrat'} \prob_s^{\maxstrat,\minstrat'}(\objvector)$. A strategy $\minstrat$ of $\plmin$ is a \emphdef{lex-optimal counter-strategy} against $\maxstrat$ if $\prob_s^\bothstrats(\objvector) = \inf_{\minstrat'} \prob_s^{\maxstrat,\minstrat'}(\objvector)$.
\end{definition}
We stress that counter-strategies of $\plmin$ depend on the strategy chosen by $\plmax$.

\subsubsection{Locally lex-optimal strategies.} An MD strategy $\maxstrat$ of $\plmax$ ($\plmin$, resp.) is called \emphdef{locally lex-optimal} if for all $s \in \Splmax$ ($s \in \Splmin$, resp.) and $a \in \Act(s)$, we have $\maxstrat(s)(a) > 0$ implies that action $a$ is lex-optimal. Thus, locally lex-optimal strategies only assign positive probability to lex-optimal actions.

\subsubsection{Convention.} For the rest of the paper, unless stated otherwise, we use  $\game = (\Splmax,\Splmin,\Act,P)$ to denote an SG and $\objvec=(\obj_1,\ldots,\obj_n)$ is a suitable (not necessarily absorbing) lex-objective, that is $\obj_i \in \{\reach{S_i}, \safe{S_i}\}$ with $S_i \subseteq S$ for all $1 \leq i \leq n$.

\section{Lexicographic SGs with Absorbing Targets}
\label{sec:absorbing}

In this section, we show how to compute the lexicographic value for SGs where \emph{all target sets are absorbing}.
We first show various theoretical results in Section~\ref{sec:MD} upon which the algorithm for computing the values and optimal strategies presented in Section~\ref{sec:algAbsorbing} is then built.
%We show in Section~\ref{sec:MD} that lex-optimal MD strategies exist in this case.
%This gives rise to an algorithm to compute the lexicographic values in Section \ref{sec:algAbsorbing}.
The main technical difficulty arises from interleaving reachability and safety objectives.
In Section \ref{sec:general}, we will reduce solving general (not necessarily absorbing) SGs to the case with absorbing targets.

\subsection{Characterizing Optimal Strategies}
\label{sec:MD}

This first subsection derives a characterization of lex-optimal strategies in terms of local optimality and an additional reachability condition (Lemma~\ref{lem:THE_lemma} further below). It is one of the key ingredients for the correctness of the algorithm presented later and also gives rise to a (non-constructive) proof of existence of MD lex-optimal strategies in the absorbing case.

We begin with the following lemma that summarizes some straightforward facts we will frequently use. Recall that a strategy is \emph{locally lex-optimal} if it only selects actions with optimal lex-value.

\begin{lemma}
	\label{lem:fact_about_strats}
	The following statements hold for any absorbing lex-objective $\objvec$:
	\begin{enumerate}[label=(\alph*)]
		\item If $\maxstrat \in \maxstratsmd$ is lex-optimal and $\minstrat \in \minstratsmd$ is a lex-optimal counter strategy against $\maxstrat$, 
		then $\maxstrat$ and $\minstrat$ are both \emph{locally} lex-optimal. (We do not yet claim that such strategies $\bothstrats$ always exist.)\label{facts:a}
%		\todo{I dropped fact b, as it is not used in main body. If we need it in Appendix, prove it there; maybe we use it in proof of fact c}
%		\item The lex-value $\lexval(s)$ of every state $s \in S$ satisfies the equations
%		\begin{align*}
%		& &\lexval(s) = \max_{a \in \Act(s)} \sum_{s'}P(s,a,s')\lexval(s') & \hspace*{5mm}\text{ if } s \in \Splmax \\
%		&\text{ and }&\lexval(s) = \min_{a \in \Act(s)} \sum_{s'}P(s,a,s')\lexval(s') & \hspace*{5mm}\text{ if } s \in \Splmin
%		\end{align*}
%		where the maximum/minimum is with respect to $\leqlex$. \label{facts:b}
		\item 
		Let $\mod{\game}$ be obtained from $\game$ by removing all actions (of both players) that are not locally lex-optimal.
		Let $\mod{\mathbf{v}}^\lexabbr$ be the lex-values in $\mod{\game}$.
		Then $\mod{\mathbf{v}}^\lexabbr = \lexval$. \label{facts:c}
		%\item For all $s \in \sinks(\game)$ it holds that $\lexval(s) = \indicvec{\objvec}{s}$, the $0$-$1$-vector whose $i$-th entry is $1$ iff $(\obj_i = \reach{S_i}$ and $s \in S_i)$ or $(\obj_i = \safe{S_i}$ and $s \notin S_i)$.
	\end{enumerate}
\end{lemma}
\begin{proof}[Sketch]
Both claims follow from the definitions of lex-value and lex-optimal strategy. For (b) in particular, we show that a strategy using actions which are not lex-optimal can be transformed into a strategy that achieves a greater (lower, resp.) value. Thus removing the non lex-optimal actions does not affect the lex-value.
See \ifarxivelse{\cite[Appendix A.1]{techreport}}{Appendix \ref{app:facts}}  for more technical details. \qed
\end{proof}

\begin{example}[Modified game $\mod{\game}$]
	Consider again the SG from Figure \ref{fig:prelimRunningA}. Recall the lex-values from Example \ref{ex:running_ex_intro}.
	Now we remove the actions that are not locally lex-optimal. This means we drop the action that leads from $p$ to $s$ and the action that leads from $r$ to $t$ or $u$ (Figure \ref{fig:prelimRunningB}).
	Since these actions were not used by the lex-optimal strategies, the value in the modified SG is the same as that of the original game.
	\qee
\end{example}

\begin{example}[Locally lex-optimal does not imply globally lex-optimal]
	\label{ex:safeEvil}
	Note that we do not drop the action that leads from $r$ to $q$, because $\lexval(r)=\lexval(q)$, so this action is locally lex-optimal.
	In fact, a lex-optimal strategy can use it arbitrarily many times without reducing the lex-value, as long as eventually it picks the action leading to $t$ or $v$.
	However, if we only played the action leading to $q$, the lex-value would be reduced to $(0,1)$ as we would not reach $S_1$, but would also avoid $S_2$. 
	
	We stress the following consequence of this: Playing a locally lex-optimal strategy is not necessarily globally lex-optimal.
	It is not sufficient to just restrict the game to locally lex-optimal actions of the previous objectives and then solve the current one. 
	Note that in fact the optimal strategy for the second objective $\safe{S_2}$ would be to remain in $\{p,q\}$; however, we must not pick this safety strategy, before we have not ``tried everything'' for all previous reachability objectives, in this case reaching $S_1$. \qee
\end{example}

This idea of ``trying everything'' for an objective $\reach{S_i}$ is equivalent to the following:
either reach the target set $S_i$, or reach a set of states from which $S_i$ cannot be reached anymore.
Formally, let $\valzeroset_i = \{ s \in S \mid \lexval_i(s) = 0 \}$ be the set of states that cannot reach the target set $S_i$ anymore.
Note that it depends on the lex-value, not the single-objective value. 
This is important, as the single-objective value could be greater than 0, but a more important objective has to be sacrificed to achieve it.
%\todo[inline]{example: Safe, reach. There is state which could achieve reach, but in doing so it violates safe. So it is in $\valzeroset_2$}

We define the set of states where we have ``tried everything'' for all reachability objectives as follows:
\begin{definition}[Final Set]
	\label{def:final_set}
	For absorbing $\objvec$, let $R_{<i} = \{j<i \mid \obj_j = \reach{S_j}\}$. We define the \emphdef{final set} $F_{<i} = \bigcup_{k \in R_{<i}} S_k\ \cup\ \bigcap_{k\in R_{<i}}\valzeroset_k$ with the convention that $F_{<i} = S$ if $R_{<i} = \emptyset$. We also let $F = F_{<n+1}$.
\end{definition}
The final set contains all target states as well as the states that have lex-value 0 for all reachability objectives; we need the intersection of the sets $\valzeroset_k$, because as long as a state still has a positive probability to reach any target set, its optimal behaviour is to try that.
\begin{example}[Final set]
%	\todo[inline]{Example: make sure you see the point I am making before.}
	\label{ex:final_set}
	For the game in Figure \ref{fig:prelimRunning}, we have $\valzeroset_1 = \{u,v,w\}$ and thus $F = \valzeroset_1 \cup S_1 = \{s,t,u,v,w\}$. 
	An MD lex-optimal strategy of $\plmax$ must almost-surely reach this set against any strategy of $\plmin$; 
	only then it has ``tried everything''. \qee
\end{example}

The following lemma characterizes MD lex-optimal strategies in terms of local lex-optimality and the final set.

\begin{lemma}
	\label{lem:THE_lemma}
	Let $\objvec$ be an absorbing lex-objective and $\maxstrat \in \maxstratsmd$.
	Then $\maxstrat$ is lex-optimal for $\objvector$ if and only if
	$\maxstrat$ is locally lex-optimal and for all $s \in S$ we have
	\begin{equation}
	\label{eq:THE_lemma}
	\forall \minstrat \in \minstratsmd\colon \prob_s^{\maxstrat,\minstrat}(\reach{F}) = 1. \tag{$\star$}
	\end{equation}
\end{lemma}
\begin{proof}[Sketch]
	The \refif-direction is shown by induction on the number $n$ of targets. We make a case distinction according to the type of $\obj_n$: If it is safety, then we prove that local lex-optimality is already sufficient for global lex-optimality. Else if $\obj_n$ is reachability, then intuitively, the additional condition \eqref{eq:THE_lemma} ensures that the strategy $\maxstrat$ indeed ``tries everything'' and either reaches the target $S_n$ or eventually a state in $\valzeroset_n$ where the opponent $\plmin$ can make sure that $\plmax$ cannot escape. The technical details of these assertions rely on a fixpoint characterization of the reachability probabilities combined with the classic Knaster-Tarski Fixpoint Theorem \cite{Tar55} and are given in \ifarxivelse{\cite[Appendix A.2]{techreport}}{Appendix \ref{app:theLemma}}.
	
	For the \refonlyif-direction recall that lex-optimal strategies are necessarily locally lex-optimal by Lemma \ref{lem:fact_about_strats} (a). Further let $i$ be such that $\obj_i = \reach{S_i}$ and assume for contradiction that $\maxstrat$ remains forever within $S \setminus (S_i \cup \valzeroset_i)$ with positive probability against some strategy of $\plmin$. But then $\maxstrat$ visits states with positive lex-value for $\obj_i$ infinitely often without ever reaching $S_i$. Thus $\maxstrat$ is not lex-optimal, contradiction.
	\qed	
\end{proof}

Finally, this characterization allows us to prove that MD lex-optimal strategies exist for absorbing objectives.
\begin{theorem}
	\label{thm:md_exists}
	For an absorbing lex-objective $\objvec$, there exist MD lex-optimal strategies for both players.
\end{theorem}
\begin{proof}[Sketch]
	We consider the subgame $\mod{\game}$ obtained by removing lex-sub-optimal actions for both players and then show that the (single-objective) value of $\reach{F}$ in $\mod{\game}$ equals $1$. 
	An optimal MD strategy for $\reach{F}$ exists \cite{Con92}; further, it is locally lex-optimal, because we are in $\mod{\game}$, and it reaches $F$ almost surely.
	Thus, it is lex-optimal for $\objvec$ by the \refif-direction of Lemma~\ref{lem:THE_lemma}.
	See \ifarxivelse{\cite[Apendix A.3]{techreport}}{Appendix \ref{app:thmMDexists}} for more details on the proof.
	\qed
\end{proof}

%Stuff that is true but probably not necessary
%Notice that with Lemma \ref{lem:THE_lemma}, it follows immediately that MD lex-optimal strategies exist when $\objvector$ has only safety objectives: Any locally lex-optimal strategy is also (globally) lex-optimal. We now show how condition \eqref{eq:THE_lemma} can be further ``condensed'' to almost sure reachability of a \emph{single} set $F$, the so-called \emphdef{Final Set}.

%Note that Lemma \ref{lem:THE_lemma} generalizes similar results that hold for single reachability/safety objectives \cite{?} \todo{TW: I'm not sure if this is actually made explicit somewhere...}. In particular, it follows immediately that MD optimal strategies of player $\plmax$ exist in the case where $\objvec$ has only safety objectives: Any locally lex-optimal strategy is already lex-optimal.

%Thus by Corollary \ref{cor:final_set_F}, the MD lex-optimal strategies for absorbing $\objvec$ are characterized as those strategies that are locally lex-optimal \emph{and} almost-surely reach the Final Set $F$ against all counter-strategies of the opponent. 
%Finally, Corollary \ref{cor:final_set_F} allows us to prove together with the known results about single-objective reachability that such strategies indeed exist:

%============================
%============================
%============================

\subsection{Algorithm for SGs with Absorbing Targets}
\label{sec:algAbsorbing}

% T: I put the last sentence from the last section here and removed the first sentence of this section, it was basically the same.
Theorem \ref{thm:md_exists} is not constructive because it relies on the values $\lexval$ without showing how to compute them. 
Computing the values and constructing an optimal strategy for $\plmax$ in the case of an absorbing lex-objective is the topic of this subsection.
%In this section, we provide an algorithm to compute the lex-values $\lexval$ and a lex-optimal MD strategy $\maxstrat$ for a given SG $\game$ and absorbing lex-objective $\objvec$.
%In order to do so, we introduce the following notion:

\begin{definition}[QRO]
	\label{def:QRO}
A \emphdef{quantified reachability objective} (QRO) is determined by a function $\quanfun \colon S' \rightarrow [0,1]$ where $S' \subseteq S$. 
For all strategies $\maxstrat$ and $\minstrat$, we define:
\[
\prob^{\bothstrats}_s(\reach{\quanfun}) = \sum_{t \in S'} \prob^{\bothstrats}_s(\until{(S \setminus S')}{t}) \cdot \quanfun(t).
\]
\end{definition}
Intuitively, a QRO generalizes its standard Boolean counterpart by additionally assigning a weight to the states in the target set $S'$.
Thus the probability of a QRO is obtained by computing the sum of the $\quanfun(t)$, $t \in S'$, weighted by the probability to avoid $S'$ until reaching $t$. 
Note that this probability does not depend on what happens after reaching $S'$; so it is unaffected by making all states in $S'$ absorbing.

In Section \ref{sec:general}, we need the dual notion of a quantified safety property, defined as $\prob^{\bothstrats}_s(\safe{\quanfun}) = 1 - \prob^{\bothstrats}_s(\reach{\quanfun})$; intuitively, this amounts to minimizing the reachability probability.

\begin{remark}
	\label{rem:qro}
	A usual reachability property $\reach{S'}$ is a special case of a quantified one with $q(s) = 1$ for all $s \in S'$. Vice versa, quantified properties can be easily reduced to usual ones defined only by the set $S'$: Convert all states $t \in S'$ into sinks, then for each such $t$ prepend a new state $t'$ with a single action $a$ and $P(t',a,t)=\quanfun(t)$ and $P(t',a,\bot)=1-\quanfun(t)$ where $\bot$ is a sink state. Finally, redirect all transitions leading into $t$ to $t'$. 
	Despite this equivalence, it turns out to be convenient and natural to use QROs.
\end{remark}

\begin{example}[QRO]
	\label{ex:qro}
	Example \ref{ex:safeEvil} illustrated that solving a safety objective after a reachability objective can lead to problems, as the optimal strategy for $\safe{S_2}$ did not use the action that actually reached $S_1$.
	In Example \ref{ex:final_set} we indicated that the final set $F = \{s,t,u,v,w\}$ has to be reached almost surely, and among those states the ones with the highest safety values should be preferred.
	This can be encoded in a QRO as follows: Compute the values for the $\safe{S_2}$ objective for the states in $F$. 
	Then construct the function $\quanfun_2 \colon F \to [0,1]$ that maps all states in $F$ to their safety value, i.e.,
	$\quanfun_2: \{s \mapsto 1,t \mapsto 0, u \mapsto 0, v \mapsto \nicefrac 1 2, w \mapsto 1\}$.
	\qee
\end{example}

Thus using QROs, we can effectively reduce (interleaved) safety objectives to quantified \emph{reachability} objectives:

\begin{lemma}[Reduction Safe $\rightarrow$ Reach]
	\label{lem:safety_reduction}
	Let $\objvec$ be an absorbing lex-objective with $\obj_n = \safe{S_n}$, $\quanfun_n \colon F \rightarrow[0,1]$ with $\quanfun_n(t) = \lexval_n(t)$ for all $t \in F$ where $F$ is the final set (Def. \ref{def:final_set}), and $\objvec'=(\obj_1,\ldots,\obj_{n-1},\reach{\quanfun_n})$. Then: $^{\objvec}\lexval~=~^{\objvec'}\lexval$.
\end{lemma}
\begin{proof}[Sketch]
	% T: mash up of what was previously in example. Was quite general, so I thought this might actually be a good proof sketch
	By definition, $^{\objvec}\lexval(s)~=~^{\objvec'}\lexval(s)$ for all $s \in F$, so we only need to consider the states in $S \setminus F$.
	Since any lex-optimal strategy for $\objvec$ or $\objvec'$ must also be lex-optimal for $\objvec_{<n}$, we know by Lemma \ref{lem:THE_lemma} that such a strategy reaches $F_{<n}$ almost-surely.
	Note that we have $F_{<n} = F$, as the $n$-th objective, either the QRO or the safety objective, does not add any new states to $F$.
	% Intuitively, for all states that are in the final set $F$, the reachability objectives are irrelevant, as we already failed or satisfied them; hence $\quanfun_n$ weighs them with their lexicographic safety value $\lexval_n$.
	%For states in $S \setminus F$, as we consider a reachability objective with respect to $F$, we ensure that an optimal strategy for $\objvec'$ maximizes the probability to reach $F$, and thus that the problem described in Example \ref{ex:safeEvil} cannot occur.
	The reachability objective $\reach{\quanfun_n}$ weighs the states in $F$ with their lexicographic safety values $\lexval_n$.
	Thus we additionally ensure that in order to reach $F$, we use those actions that give us the best safety probability afterwards.
	In this way we obtain the correct lex-values $\lexval_n$ even for states in $S \setminus F$.
	See \ifarxivelse{\cite[Appendix A.4]{techreport}}{Appendix \ref{app:safetyRed}} for the full technical proof.
%	\todo{M: drop or improve this sketch? Right now, it's mainly the intuition conveyed before.}
%	Recall that $\reach{\quanfun_n}$ is a quantified \emph{reachability} objective whose semantics is defined as
%	\[
%	\prob_s^\bothstrats(\reach{\quanfun_n}) = \sum_{t \in F} \prob_s^\bothstrats(\until{(S \setminus F)}{t}) \cdot \lexval_n(t)
%	\]
%	for all strategies $\bothstrats$ and states $s \in S$. 
%	Since any lex-optimal strategy $\maxstrat$ for $\objvec$ or $\objvec'$ must also be lex-optimal for $\objvec_{<n}$, we know by Lemma \ref{lem:THE_lemma} that $\maxstrat$ reaches $F_{<n} = F$ almost-surely against all $\minstrat \in \minstratsmd$. 
%	Thus $\lexval_n(s)$ for $s \notin F$ equals the sum of $\lexval_n(t)$, $t \in F$, weighted with the probability of a lex-optimal strategy to reach this $t$ from $s$ without seeing a state in $F$ before.
%	Thus a lex-optimal strategy must maximize precisely the QRO $\reach{\quanfun_n}$.
	%See appendix for more technical details.%~\ref{app:safetyRed}.
	\qed
\end{proof}

%\proofsafetyreduction

\begin{example}[Reduction Safe $\rightarrow$ Reach]
	Recall Example \ref{ex:qro}. By the preceding Lemma \ref{lem:safety_reduction},
	computing $\sup_{\maxstrat} \inf_{\minstrat} \prob^{\bothstrats}_s(\reach{S_1}, \reach{\quanfun_2})$ yields the correct lex-value $\lexval(s)$ for all $s \in S$.
	%See Figure~\ref{fig:illustration_F} for an illustration of this idea.\todo{M: write more or drop; I like the figure, but it doesn't fit my flow right now}
	Consider for instance state $r$ in the running example: The action leading to $q$ is clearly suboptimal for $\reach{\quanfun_2}$ as it does not reach $F$.
	Both other actions surely reach $F$. 
	However, since $\quanfun_2(t) = \quanfun_2(u) = 0$ while $\quanfun_2(v) = \nicefrac 1 2$, the action leading to $u$ and $v$ is preferred over that leading to $t$ and $u$, as it ensures the higher safety probability after reaching $F$.
	\qee
\end{example}

\begin{algorithm}[h]
	\caption{Solve absorbing lex-objective}
	\label{alg:absorbing}
	\begin{algorithmic}[1]
		\Require{ SG $\game$, absorbing lex-objective $\objvec = (\obj_1, \dots, \obj_n)$}
		\Ensure{Vector of lex-values $\lexval$, MD lex-optimal strategy $\maxstrat$ for $\plmax$}
		\Procedure{$\algabsorbing$}{$\game,\objvec$}
		\State initialize $\lexval$ and $\maxstrat$ arbitrarily
		\State $\mod{\game} \gets \game$ \Comment{Consider whole game in the beginning.}
		
		\medskip
		
		\For{$1 \leq i \leq n$}
		\State $(\val, \mod{\maxstrat}) \gets \algsingleobj(\mod{\game},\obj_i)$ \label{line:startSingleObj}
		\If{$\obj_i = \safe{S_i}$} \label{line:ifSafe}
		\State $F_{<i} \gets$ final set with respect to $\mod{\game}$ and $\objvec_{<i}$ \Comment{see Def. \ref{def:final_set}}
		\State $\quanfun_i(s) \gets \val(s)$ for all $s \in F_{<i}$ \Comment{see Def. \ref{def:QRO}} 
		\State $(\val, \maxstrat_{Q}) \gets \algsingleobj(\mod{\game},\reach{\quanfun_i})$ \label{line:qro}
		\EndIf \label{line:endSingleObj}
		
		\medskip
		 
		\State $\mod{\game} \gets $ restriction of $\mod{\game}$ to optimal actions w.r.t. $\val$ \label{line:restrict}
		
		\medskip
		
		\State $\lexval_i \gets \val$ \label{line:startResult}
		\For {$s \in S$}
		\If {($\obj_i = \reach{S_i}$ and $\val(s) > 0$) or ($\obj_i = \safe{S_i}$ and $s \in F_{<i}$)}
			\State $\maxstrat(s) \gets \mod{\maxstrat}(s)$ \Comment{Strategy improvement}
		\Else {~\textbf{if} $\obj_i = \safe{S_i}$ and $s \notin F_{<i}$}
			\State $\maxstrat(s) \gets \maxstrat_{Q}(s)$
		\EndIf
		\EndFor \label{line:endResult}
		
		\EndFor
		
		\Return $(\lexval, \maxstrat)$
		\EndProcedure
	\end{algorithmic}
\end{algorithm}

We now explain the basic structure of Algorithm \ref{alg:absorbing}. 
More technical details are explained in the proof sketch of Theorem \ref{thm:alg_correct} and the full proof is in \ifarxivelse{\cite[Appendix A.5]{techreport}}{Appendix \ref{app:algAbsCorr}}.
The idea of Algorithm \ref{alg:absorbing} is, as sketched in Section~\ref{sec:MD}, to consider the objectives sequentially in the order of importance, i.e., starting with $\obj_1$.
The $i$-th objective is solved (Lines \ref{line:startSingleObj}-\ref{line:endSingleObj}) and the game is restricted to only the locally optimal actions (Line \ref{line:restrict}).
This way, in the $i$-th iteration of the main loop, only actions that are locally lex-optimal for objectives 1 through $(i{-}1)$ are considered.
Finally, we construct the optimal strategy and update the result variables (Lines~\ref{line:startResult}-\ref{line:endResult}).

\begin{theorem}
	\label{thm:alg_correct}
	Given an SG $\game$ and an absorbing lex-objective $\objvec = (\obj_1, \dots, \obj_n)$, 
	Algorithm \ref{alg:absorbing} correctly computes the vector of lex-values $\lexval$ and an MD lex-optimal strategy $\maxstrat$ for player $\plmax$.
	It needs $n$ calls to a single objective solver.
\end{theorem}

\begin{proof}[Sketch]
	%We now explain the intuition of the algorithm and highlight some interesting details. See Appendix \ref{app:algAbsCorr} for the formal proof of correctness.
	%We only highlight the most important aspects of the algorithm, see the appendix for more formal details.
\begin{itemize}
	\item \textbf{$\mod{\game}$-invariant:} For $i>1$, in the $i$-th iteration of the loop, $\mod{\game}$ is the original SG restricted to only those actions that are locally lex-optimal for the targets 1 to $(i{-}1)$; this is the case because Line \ref{line:restrict} was executed for all previous targets.
	\item \textbf{Single-objective case:} The single-objective that is solved in Line~\ref{line:startSingleObj} can be either reachability or safety. 
	We can use any (precise) single-objective solver as a black box, e.g. strategy iteration~\cite{HK66}. Recall that by Remark \ref{rem:qro}, it is no problem to call a single-objective solver with a QRO since there is a trivial reduction.
	\item \textbf{QRO for safety:} If an objective is of type reachability, no further steps need to be taken; 
	if on the other hand it is safety, we need to ensure that the problem explained in Example \ref{ex:safeEvil} does not occur. 
	Thus we compute the final set $F_{<i}$ for the $i$-th target and then construct and solve the QRO as in Lemma~\ref{lem:safety_reduction}.
	%This means we need to ensure that the resulting strategy tries to achieve every previous (more important) reachability objective before using the safety strategy.
	%Thus, we compute the final set $F$ with respect to $\mod{\game}$ all previous objective $\objvec_{<i}$, and then set up a QRO as follows:
	%the target set is $F$ and the values assigned to the targets are the safety values that we computed in Line \ref{line:startSingleObj}.
	%By solving the QRO, we ensure that a strategy can only have a positive value if it first reaches $F$, i.e. tries to achieve all previous reachability objectives. 
	%By weighing the states in $F$ with their safety value, we ensure that the strategy prefers those states in $F$ where it has a higher chance of satisfying the safety objective.
	%Note that, as we work on $\mod{\game}$, it is irrelevant which state in $F$ we try to reach, as every strategy reaching them can only use lex-optimal actions for the previous objectives by the $\mod{\game}$-invariant.\todo{T: this point essentially repeats what was previously said. we could just reference Lemma 3 here}
	\item \textbf{Resulting strategy:} When storing the resulting strategy, we again need to avoid errors induced by the fact that locally lex-optimal actions need not be globally lex-optimal.
	This is why for a reachability objective, we only update the strategy in states that have a positive value for the current objective; if the value is 0, the current strategy does not have any preference, and we need to keep the old strategy.
%	For safety objectives, we need to update the strategy in two ways:
%	for all states in the final set $F_{<i}$ we set it to the safety strategy $\mod{\maxstrat}$ from Line~\ref{line:startSingleObj}.
%	For all states in $S \setminus F_{<i}$, we set it to the reachability strategy from the QRO $\maxstrat_Q$ from Line~\ref{line:qro}.
	For safety objectives, we need to update the strategy in two ways: for 
	all states in the final set $F_{<i}$, we set it to the safety strategy 
	$\mod{\maxstrat}$ (from Line~\ref{line:startSingleObj}) as within $F_{<i}$ we do not have to consider 
	the previous reachability objectives and therefore must follow an 
	optimal safety strategy. For all states in $S \setminus F_{<i}$, we set it to the 
	reachability strategy from the QRO $\maxstrat_Q$ (from Line\ref{line:qro}). This is 
	correct, as $\maxstrat_Q$ ensures almost-sure reachability of $F_{<i}$ which is 
	necessary to satisfy all preceding reachability objectives; moreover 
	$\maxstrat_Q$ prefers those states in $F_{<i}$ that have a higher safety value 
	(cf. Lemma~\ref{lem:safety_reduction}).
	
	\item \textbf{Termination:} The main loop of the algorithm invokes $\algsingleobj$ for each of the $n$ objectives.
\end{itemize}
\qed
\end{proof}

%\subsubsection{Proof of correctness}
%\begin{proof}[of Theorem \ref{thm:alg_correct} -- Sketch]
%	\todo{reread; also probably refer to the description of algo, as that intuitively justifies many of the ideas.}
%	In the case $\obj_n = \reach{S_n}$ (line \ref{line:ifreach}), one can prove using Lemma \ref{lem:THE_lemma} that the MD strategy $\maxstrat$ which is constructed as a modification of $\maxstrat_{<n}$ in lines \ref{line:construct_strat_start} - \ref{line:construct_strat_end} remains lex-optimal for the first $n-1$ objectives $\objvec_{<n}$. Then we argue that the values $\mod{\val}$ are indeed equal to $\lexval_n$, the $n$-th entries of the lex-values. The correctness of the treatment of the other case $\obj_n = \safe{S_n}$ (line \ref{line:ifsafe}) follows from Lemma \ref{lem:safety_reduction}. The full proof is given in the appendix.
%	\qed
%\end{proof}

%\paragraph{Remark} We explicitly allowed quantified reachability objectives $\reach{q}$ in the input to $\algabsorbing$. It is not difficult to see that the algorithm can actually also handle quantified \emph{safety} objectives $\safe{q}$ since they can be reduced to a standard safety objective $\safe{S_q}$ is the same way as reachability objectives. This feature is needed later in Section \ref{sec:non_absorbing}

\section{General Lexicographic SGs}
\label{sec:general}

We now consider $\objvec$ where $S_i \subseteq \sinks(\game)$ does \emph{not} necessarily hold.
Section \ref{sec:non_absorbing} describes how we can reduce these general lex-objectives to the absorbing case. 
The resulting algorithm is given in Section \ref{sec:algGeneral} and the theoretical implications in Section \ref{sec:theorImpl}.

\subsection{Reducing General Lexicographic SGs to SGs with Absorbing Targets}
\label{sec:non_absorbing}

In general lexicographic SG, strategies need memory, because they need to remember which of the $S_i$ have already been visited and behave accordingly. 
We formalize the solution of such games by means of \emphdef{stages}. 
Intuitively, one can think of a stage as a copy of the game with less objectives, 
or as the sub-game that is played after visiting some previously unseen set $S_i$. 
\begin{definition}[Stage]
Given an arbitrary lex-objective $\objvec = (\objvec_1, \dots, \objvec_n)$ and a set $I \subseteq \{i \leq n\}$, a \emphdef{stage} $\objvec(I)$ is the objective vector where the objectives $\objvec_i$ are removed for all $i \in I$.

For state $s \in S$, let $\objvec(s)=\objvec(\{i \mid s \in S_i\})$.
If a stage contains only one objective, we call it \emph{simple}.
\end{definition}

%A stage is uniquely determined by a sub-objective $\objvec'$ which only contains the $\obj_i$ for which $S_i$ has not yet been visited. 

%Formally, for all $j\leq n$ let $\objvec_{j}$ be the new stage that begins upon reaching $s \in S_j$.
%It is obtained by removing the $j$-th objective from $\objvec$.
%\todo{trivial stage was earlied defined to only consist of sinks, and the absorbing case was called trivial. However, stage is a vector of objectives, so this does not fit together. The earlier definition implicitly also used the starting state s, but this would make the algorithm more complex. Hence I chose to define it a bit differently (now called simple), but in the spirit of the algo.}
%A stage is \emphdef{trivial} if it only consists of a sink or if all $n$ sets $S_i$ have been visited. Notice that $\game$ can be decomposed into at most $2^n - 1$ non-trivial stages and that for absorbing $\objvec$, there is just one non-trivial stage (which is the reason why we started with it).

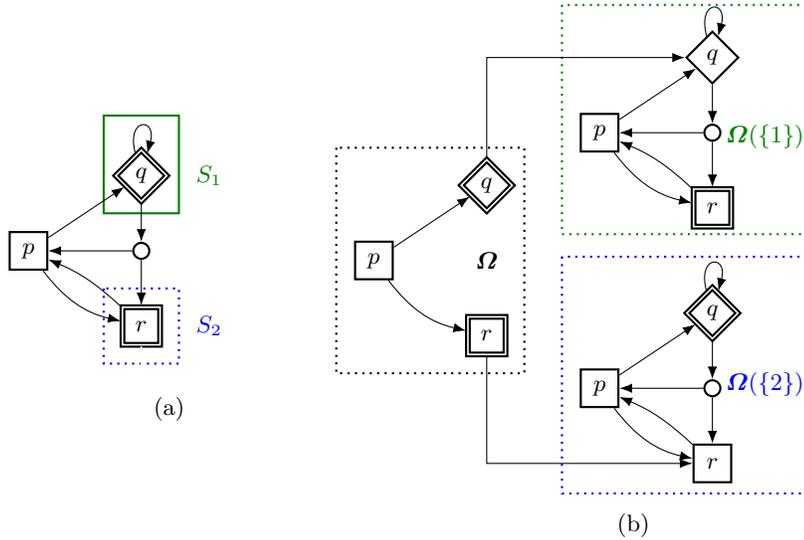
\begin{figure}[t]
	\begin{subfigure}{0.35\textwidth}
		\begin{tikzpicture}
		\node[max] (s0) at (-0.5,0) {$p$};
		\node[prob] (p) at (1,0){};
		\node[min,target] (s1) at (1,1){$q$};
		\node[max,target] (s2) at (1,-1) {$r$};
		
		\draw[green!50!black,thick] ($(s1)+(-0.5,0.8)$)  rectangle ($(s1)+(0.5,-0.5)$);
		\draw[blue,thick,dotted]     ($(s2)+(-0.5,0.5)$) rectangle ($(s2)+(0.5,-0.5)$);
		\node (t1) at (1.9,1) {\textcolor{green!50!black}{$S_1$}};
		\node (t2) at (1.9,-1) {\textcolor{blue}{$S_2$}};

		\draw[trans] (s0) -- (s1);
		\draw[trans] (s0) edge[bend right=20] (s2);
		\draw[trans] (s1) edge[loop above] (s1);
		\draw[trans] (s1) -- (p);
		\draw[trans] (s2) edge[bend right=10] (s0);
		\draw[trans] (p) -- (s0);
		\draw[trans] (p) -- (s2);
		
		\draw[white] (s2) -- ++(0,-7mm) -| (s2); % this is a space hack so that subfigs align
		\end{tikzpicture}
		\subcaption{}
		\label{fig:unfoldingA}
	\end{subfigure}
	\begin{subfigure}{0.65\textwidth}
		\begin{tikzpicture}
		
		\draw[black,thick,dotted]     (-1,1.5) rectangle (1.5,-1.5);
		\draw[green!50!black,thick,dotted]     (2,3.4) rectangle (5.3,0.35);
		\draw[blue,thick,dotted]     (2,0.05) rectangle (5.3,-3.1);
		\node (o) at (1,0) {$\objvec$};
		\node (o1) at (4.7,1.65) {\textcolor{green!50!black}{$\objvec(\{1\})$}};
		\node (o1) at (4.7,-1.65) {\textcolor{blue}{$\objvec(\{2\})$}};

		\node[max] (s0) at (-0.5,0) {$p$};
		%\node[prob] (p) at (1,0){};
		\node[min,target] (s1) at (1,1){$q$};
		\node[max,target] (s2) at (1,-1) {$r$};
		
		%+3,+1.7
		\node[max] (s0a) at (2.5,1.7) {$p$};
		\node[prob] (pa) at (4,1.7){};
		\node[min] (s1a) at (4,2.7){$q$};
		\node[max,target] (s2a) at (4,0.7) {$r$};
		
		%+3,-1.7
		\node[max] (s0b) at (2.5,-1.7) {$p$};
		\node[prob] (pb) at (4,-1.7){};
		\node[min,target] (s1b) at (4,-0.7){$q$};
		\node[max] (s2b) at (4,-2.7) {$r$};
		
		\draw[trans] (s0) -- (s1);
		\draw[trans] (s0) edge[bend right=20] (s2);
		\draw[trans] (s1) |- (s1a);
		\draw[trans] (s2) |- (s2b);
		%\draw[trans] (p) -- (s0);
		%\draw[trans] (p) -- (s2);
		
		\draw[trans] (s0a) -- (s1a);
		\draw[trans] (s0a) edge[bend right=20] (s2a);
		\draw[trans] (s1a) edge[loop above] (s1a);
		\draw[trans] (s1a) -- (pa);
		\draw[trans] (s2a) edge[bend right=10] (s0a);
		\draw[trans] (pa) -- (s0a);
		\draw[trans] (pa) -- (s2a);
		
		\draw[trans] (s0b) -- (s1b);
		\draw[trans] (s0b) edge[bend right=20] (s2b);
		\draw[trans] (s1b) edge[loop above] (s1b);
		\draw[trans] (s1b) -- (pb);
		\draw[trans] (s2b) edge[bend right=10] (s0b);
		\draw[trans] (pb) -- (s0b);
		\draw[trans] (pb) -- (s2b);

		\end{tikzpicture}
		\subcaption{}
		\label{fig:unfoldingB}
	\end{subfigure}
	\caption{(a) SG with non-absorbing lex-objective $\objvec = (\reach{S_1}, \reach{S_2})$. (b) The three stages identified by the sub-objectives $\objvec$, $\objvec(\{1\}) = (\reach{S_2})$ and $\objvec(\{2\}) = (\reach{S_1})$. The two stages on the right are both \emph{simple}.}
	\label{fig:unfolding}
\end{figure}

\begin{example}[Stages]
	Consider the SG in Figure \ref{fig:unfoldingA}.
	As there are two objectives, there are four possible stages: The one where we consider both objectives (the region denoted with $\objvec$ in Figure \ref{fig:unfoldingB}), the \emph{simple} ones where we consider only one of the objectives (regions $\objvec(\{1\})$ and $\objvec(\{2\})$), and the one where both objectives have been visited. 
	The last stage is trivial since there are no more objectives, hence we do not depict it and do not have to consider it.
	The actions of $q$ and $r$ are omitted in the $\objvec$-stage, as upon visiting these states, a new stage begins.
	
	Consider the simple stages: in stage $\objvec(\{1\})$, $q$ has value 0, as it is a $\plmin$-state and will use the self-loop to avoid reaching $r \in S_2$.
	In stage $\objvec(\{2\})$, both $p$ and $r$ have value 1, as they can just go to the target state $q \in S_1$.
	Combining this knowledge, we can get an optimal strategy for every state.
	In particular, note that an optimal strategy for state $p$ needs memory: First go to $r$ and thereby reach stage $\objvec(\{2\})$. Afterwards, go from $r$ to $p$ and now, on the second visit in a different stage, use the other action in $p$ to reach $q$.
	In this example, we observe another interesting fact about lexicographic games: it can be optimal to first satisfy less important objectives. \qee
\end{example}
%Consider a stage identified with lex-objective $\objvec = (\obj'_1,\ldots,\obj'_m)$, $0 < m \leq n$, and

In the example, we combined our knowledge of the sub-stages to find the lex-values for the whole lex-objective.
In general, the values for the stages are numbers in $[0,1]$.
Thus we reuse the idea of \emph{quantified} reachability and safety objectives, see Definition \ref{def:QRO}.

For all $1 \leq i \leq n$, let $\quanfun_i \colon \bigcup_{j\leq n} S_j \to [0,1]$ by defined by:
\[
\quanfun_i(s) = \begin{cases}
1 &$if $ s \in S_i $ and else:$\\
\phantom{1 - .} ^{\objvec(s)}\lexval_{i}(s) &$if $ \obj_i $ is reachability$ \\
1 -\ ^{\objvec(s)}\lexval_{i}(s) & $if $ \obj_i $ is safety.$
\end{cases}
\]
To keep the correct type of every objective, we let 
$\qobjvec = (\type_1(\quanfun_1),\ldots,\type_n(\quanfun_n))$ where for all $1 \leq i \leq n$, $\type_i = \mathsf{Reach}$ if $\obj_i=\reach{S_i}$ and else $\type_i = \mathsf{Safe}$ if $\obj_i = \safe{S_i}$.
So we have now reduced a general lexicographic objective $\objvec$ to a vector of quantitative objectives $\qobjvec$.
Lemma \ref{lem:non_absorbing_eq} shows that this reduction preserves the values.
\begin{lemma}
	\label{lem:non_absorbing_eq}
	For arbitrary lex-objectives $\objvec$ it holds that $^{\objvec}\lexval =\ ^{\qobjvec}\lexval$.
\end{lemma}
\begin{proof}[Sketch]
	We write $\alltargets = \bigcup_{j\leq n} S_j$ for the sake of readability in this sketch.
	By induction on the length $n$ of the lex-objective $\objvec$, it is easy to show that the equation holds in states $s \in \alltargets$, i.e., $^{\objvec}\lexval(s) =\ ^{\qobjvec}\lexval(s)$.
	For a state $s$ which is not contained in any of the $S_j$, and for any strategies $\bothstrats$ we have the following equation 
	\begin{align*}
	\prob^\bothstrats_s(\reach{S_i}) &= \sum_{\path t \in Paths_{fin}(\alltargets)} \prob_s^\bothstrats(\path t) \cdot \prob_{\path t}^{\bothstrats}(\reach{S_i})
	\end{align*}
	where $Paths_{fin}(\alltargets) = \{\path t \in ((S \setminus \alltargets)\times \actlabels)^* \times S \mid  t \in \alltargets\}$ denotes the set of all finite paths to a state in $\alltargets$ in the Markov chain $\game^\bothstrats$ and $\prob_s^\bothstrats(\path t)$ is the probability of such a path when $\game^\bothstrats$ starts in $s$. From this we deduce that in order to maximize the left hand size of the equation in the lexicographic order, we should play such that we prefer reaching states in $\alltargets$ where $\quanfun_i$ has a higher value; that is, we should maximize the QRO $\reach{\quanfun_i}$. The argument for safety is similar and detailed in \ifarxivelse{\cite[Appendix A.6]{techreport}}{Appendix \ref{app:redGenAbsorb}}.
	\qed
\end{proof}

The functions $\quanfun_i$ involved in $\qobjvec$ \emph{all have the same domain} $\bigcup_{j\leq n} S_j$.
%Thus, upon reaching any target state, the quantitative objective returns the lex-value for the whole objective.
Hence we can, as mentioned below Definition \ref{def:QRO}, consider $\qobjvec$ on the game where all states in $\bigcup_{j\leq n} S_j$ are sinks without changing the lex-value.
This is precisely the definition of an absorbing game, and hence we can compute $^{\qobjvec}\lexval$ using Algorithm \ref{alg:absorbing} from Section \ref{sec:algAbsorbing}.

\subsection{Algorithm for General SG}
\label{sec:algGeneral}

Algorithm \ref{alg:general} computes the lex-value $^{\objvec}\lexval$ for a given lexicographic objective $\objvec$ and an arbitrary SG $\game$. We highlight the following technical details:
\begin{itemize}
\item \textbf{Reduction to absorbing case:} We just have seen, that once we have the quantitative objective vector $\qobjvec$, we can use the algorithm for absorbing SG (Line \ref{line:callAbsorb}).
\item \textbf{Computing the quantitative objective vector:} To compute $\qobjvec$, the algorithm calls itself recursively on all states in the union of all target sets (Line \ref{line:recurseStart}-\ref{line:recurseEnd}). We annotated this recursive call ``With dynamic programming'', as we can reuse the results of the computations. In the worst case, we have to solve all $2^n - 1$ possible non-empty stages.
Finally, given the values $^{\objvec(s)}\lexval$ for all $s \in \bigcup_{j\leq n} S_j$, we can construct the quantitative objective (Line \ref{line:quanFun} and \ref{line:quanObj}) that is used for the call to $\algabsorbing$.
\item \textbf{Termination:} Since there are finitely many objectives in $\objvec$ and in every recursive call at least one objective is removed from consideration, eventually we have a \emph{simple} objective that can be solved by $\algsingleobj$~(Line~\ref{line:simple}). 
\item \textbf{Resulting strategy:} The resulting strategy is composed in Line \ref{line:returnStrat}: It adheres to the strategy for the quantitative query $^{\qobjvec}\maxstrat$ until some $s \in \bigcup_{j \leq n} S_j$ is reached. 
Then, to achieve the values promised by $q_i(s)$ for all $i$ with $s \notin S_i$, it adheres to $^{\objvec(s)} \maxstrat$, the optimal strategy for stage $\objvec(s)$ obtained by the recursive call.
\end{itemize}

\begin{algorithm}[t]
	\caption{Solve general lex-objective}
	\label{alg:general}
	\begin{algorithmic}[1]
		\Require SG $\game$, lex-objective $\objvec = (\obj_1, \dots, \obj_n)$ 
		\Ensure Lex-values $^{\objvec}\lexval$, lex-optimal $\maxstrat \in \maxstrats$ with memory of class-size~$\leq2^n-1$
		\Procedure{$\alggeneral$}{$\game,\objvec$}
		\If {$\objvec$ is \emph{simple}}
			\State \Return $\algsingleobj(\game,\obj_1)$ \label{line:simple}
		\EndIf
		
		\medskip
		
		\For{$s \in \bigcup_{j \leq n} S_j$}\label{line:recurseStart}
			\State $\left(^{\objvec(s)}\lexval,\ ^{\objvec(s)}\maxstrat \right) \gets \alggeneral(\game,\objvec(s))$ \Comment{With dynamic programming} \label{line:dyn_prog_call}
		\EndFor \label{line:recurseEnd}
		\For{$1 \leq i\leq n$}
			\State Let $\quanfun_i \colon \bigcup_{j\leq n} S_j \rightarrow [0,1]$,\ $\quanfun_i(s) \gets
			\begin{cases}
			1 &$if $ s \in S_i $ and else:$ \\
			~~~~~\ ^{\objvec(s)}\lexval_{i}(s) &$if $\type(\obj_i) = \mathsf{Reach} \\
			1 -\ ^{\objvec(s)}\lexval_{i}(s) & $if $\type(\obj_i) = \mathsf{Safe}
			\end{cases}$
			\label{line:quanFun}
		\EndFor		
		\State $\qobjvec \gets (\type_1(\quanfun_1),\ldots,\type_n(\quanfun_n))$ \label{line:quanObj}
		
		\medskip
		
		\State $(^{\qobjvec}\lexval,\ ^{\qobjvec}\maxstrat) \gets \algabsorbing(\game, \qobjvec)$\label{line:callAbsorb}
		\State $\maxstrat \gets$ adhere to $^{\qobjvec}\maxstrat$ until some $s \in \bigcup_{j \leq n} S_j$ is reached. Then adhere to $^{\objvec(s)}\maxstrat$. \label{line:returnStrat}
		\State \Return $(^{\qobjvec}\lexval, \maxstrat)$
		\EndProcedure
	\end{algorithmic}
\end{algorithm}

\begin{corollary}
	Given an SG $\game$ and an arbitrary lex-objective $\objvec = (\obj_1,\dots,\obj_n)$, 
	Algorithm \ref{alg:general} correctly computes the vector of lex-values $\lexval$ and a deterministic lex-optimal strategy $\maxstrat$ of player $\plmax$ which uses memory of class-size $\leq 2^n -1$.
	The algorithm needs at most $2^n -1$ calls to $\algabsorbing$ or $\algsingleobj$.
\end{corollary}
\begin{proof}
	Correctness of the algorithm and termination follows from the discussion of the algorithm, Lemma \ref{lem:non_absorbing_eq} and Theorem \ref{thm:alg_correct}. 
	\qed
\end{proof}

%\todo[inline]{TW: still incomplete... The main idea here is to solve each stage individually by viewing it as a game with an absorbing objective. The only technical difference to the previous chapter is that the objectives can now consist of arbitrary quantified safety and reachability objectives (see Def. in Preliminaries). These quantified objectives are obtained (recursively) by computing the values of the next stage. However, I think that the algorithm extends immediately to those objectives. @Maxi: this is the same what I did in the thesis, I just didn't call the subgames ``stages'' there.}

\subsection{Theoretical Implications: Determinacy and Complexity}
\label{sec:theorImpl}

Theorem \ref{cor:absorbing_determined} below states that lexicographic games are \emph{determined} for arbitrary lex-objectives $\objvec$. 
Intuitively, this means that the lex-value is independent from the player who fixes their strategy first. 
%JP suggested its instead of their, but their is the current standard https://english.stackexchange.com/questions/30455/is-using-he-for-a-gender-neutral-third-person-correct
Recall that this property does not hold for non-lexicographic multi-reachability/safety objectives~\cite{DBLP:conf/mfcs/ChenFKSW13}. 

\begin{theorem}[Determinacy]
	\label{cor:absorbing_determined}
	For general SG $\game$ and lex-objective $\objvec$, it holds for all $s \in S$ that:
	\[
		\lexval(s) = \adjustlimits \sup_{\maxstrat} \inf_{\minstrat } \prob_s^\bothstrats(\objvec) = \adjustlimits \inf_{\minstrat } \sup_{\maxstrat} \prob_s^\bothstrats(\objvec).
	\]
\end{theorem}
\begin{proof}
	This statement follows because single-objective games are determined \cite{Con92} and Algorithm \ref{alg:general} obtains all values by either solving single-objective instances directly (Line \ref{line:simple}) or calling Algorithm \ref{alg:absorbing}, which also reduces everything to the single-objective case (Line \ref{line:startSingleObj} of Algorithm \ref{alg:absorbing}). 
	Thus the sup-inf values $\lexval$ returned by the algorithm are in fact equal to the inf-sup values.
	\qed
\end{proof}

By analyzing Algorithm \ref{alg:general}, we also get the following complexity results:

\begin{theorem}[Complexity]
	For any SG $\game$ and lex-objective $\objvec = (\obj_1, \dots, \obj_n)$:
	\begin{enumerate}
		\item {\em Strategy complexity:} Deterministic strategies with $2^n - 1$ memory-classes (i.e., bit-size $n$) are sufficient and necessary for lex-optimal strategies.
		\item {\em Computational complexity:} The lex-game decision problem ($\lexval(s_0) \geqlex$ $\vec{x}$?) is $\pspace$-hard and can be solved in $\nexptime \cap \conexptime$. If $n$ is a constant or $\objvec$ is absorbing, then it is contained in $\np \cap \conp$.
	\end{enumerate}
\end{theorem}
\begin{proof}
	\begin{enumerate}
		\item For each stage, Algorithm \ref{alg:general} computes an MD strategy for the quantitative objective. These strategies are then concatenated whenever a new stage is entered.
		Equivalently, every stage has an MD strategy for every state, so as there are at most $2^n - 1$ stages (since there are $n$ objectives), the strategy needs at most $2^n - 1$ states of memory; these can be represented with $n$ bits. Intuitively, we save for every target set whether it has been visited.
		The memory lower bound already holds in non-stochastic reachability games where all $n$ targets have to be visited with certainty \cite{FH10}.
		\item The work of \cite{DBLP:journals/fmsd/RandourRS17} shows that in MDPs, it is $\pspace$-hard to decide if $n$ targets can be visited almost-surely. This problem trivially reduces to ours. For the $\np$ upper bound, observe that there are at most $2^n - 1$ stages, i.e., a constant amount if $n$ is assumed to be constant (or even just one stage if $\objvec$ is absorbing). Thus we can guess an MD strategy for player $\plmax$ in every stage. The guessed overall strategy can then be checked by analyzing the induced MDP in polynomial time \cite{DBLP:journals/lmcs/EtessamiKVY08}. The same procedure works for player $\plmin$ and since the game is determined, we have membership in $\conp$. In the same way we obtain the $\nexptime \cap \conexptime$ upper bound in the general case where $n$ is arbitrary.
		\qed
	\end{enumerate}
\end{proof}

%Moreover, $\expspace$ is easily seen to be an upper bound for the decision problem in the general case.
We leave the question whether $\pspace$ is also an upper bound open. The main obstacle towards proving $\pspace$-membership is that it is unclear if the lex-value -- being dependent on the value of \emph{exponentially} many stages in the worst-case -- may actually have exponential bit-complexity.

\section{Experimental Evaluation}
\label{sec:exp}

In this section, we report the results of a series of experiments made with a prototypical implementation of our algorithm.

%\todo{
%Point of this section:
%1. There are cases where lexicographic preferences make sense, see case studies.
%2. The algorithm works, i.e. produces the expected result in reasonable time. Also: first such algo, hence no competitor to compare to. Still, we do the following:
%3. Compare time to single obj (first obj/all obj), to show that it is kinda reasonable time. (Note that all obj single needs the loading several times, we should exclude that)
%4. Some analysis of the algo: Stages, number of actions over iterations of main algo.
%}

\subsubsection{Case Studies.}

We have considered the following case studies for our experiments:
\begin{description}
	\item[Dice] This example is shipped with PRISM-games~\cite{DBLP:journals/sttt/KwiatkowskaPW18} and models a simple dice game between two players.
	%First player 1 throws a fair die repeatedly until accepting an outcome. Up to $N$ tries are possible. Then player 2 is allowed to throw  the die at most as many times as player 1 did.
	%A player wins if it achieves a higher outcome than the opponent; draws are possible.
	The number of throws in this game is a configurable parameter, which we instantiate with 10, 20 and 50.
	The game has three possible outcomes: Player $\plmax$ wins, Player $\plmin$ wins or draw.
	A natural lex-objective is thus to maximize the winning probability and then the probability of a draw.
	\item[Charlton] This case study \cite{DBLP:conf/qest/ChenKSW13} is also included in PRISM-games. It models an autonomous car navigating through a road network. A natural lex-objective is to minimize the probability of an accident (possibly damaging human life) and then maximize the probability to reach the destination.
	\item[Hallway (HW)] This instance is based on the Hallway example standard in the AI literature \cite{LCK95,CCGK16}. A robot can move north, east, south or west in a known environment, but each move only succeeds with a certain probability and otherwise rotates or moves the robot in an undesired direction. We extend the example by a target wandering around based on a mixture of probabilistic and demonic non-deterministic behavior, thereby obtaining a stochastic game modeling for instance a panicking human in a building on fire. Moreover, we assume a 0.01 probability of damaging the robot when executing certain movements; the damaged robot's actions succeed with even smaller probability. The primary objective is to save the human and the secondary objective is to avoid damaging the robot. We use square grid-worlds of sizes  5$\times$5, 8$\times$8 and 10$\times$10.
	\item[Avoid the Observer (AV)] This case study is inspired by a similar example in \cite{CC15}. It models a game between an intruder and an observer in a grid-world. 
	The grid can have different sizes as in HW, and we use 10$\times$10, 15$\times$15 and 20$\times$20.
	The most important objective of the intruder is to avoid the observer, its secondary objective is to exit the grid. We assume that the observer can only detect the intruder within a certain distance and otherwise makes random moves. At every position, the intruder moreover has the option to stay and search to find a precious item. In our example, this occurs with probability 0.1 and is assumed to be the third objective. 
%	\item[Rock Sampling (RS)] \todo{todo? or leave it?}
\end{description}

\subsubsection{Implementation and Experimental Results.}

We have implemented our algorithm within PRISM-games~\cite{DBLP:journals/sttt/KwiatkowskaPW18}. Since PRISM-games does not provide an \emph{exact} algorithm to solve SGs, we used the available value iteration to implement our single-objective blackbox. Note that since this value iteration is not exact for single-objective SGs, we cannot compute the exact lex-values. Nevertheless, we can still measure the overhead introduced by our algorithm compared to a single-objective solver.

In our implementation, value iteration stops if the values do not change by more than $10^{-8}$ per iteration, which is PRISM's default configuration. The experiments were conducted on a 2.4 GHz Quad-Core Intel\textsuperscript{\textcopyright} Core\texttrademark\ i5 processor, with 4GB of RAM available to the Java VM.
The results are reported in Table \ref{tab:results}. We only recorded the run time of the actual algorithms; the time needed to parse and build the model is excluded. All numbers are rounded to full seconds. All instances (even those with state spaces of order $10^6$) could be solved within a few minutes.
 
\begin{table}[t]
	\caption{Experimental Results. The two leftmost columns of the table show the type of the lex-objective, the name of the case studies, possibly with scaling parameters, and the number of states in the model. The next three columns give the verification times (excluding time to parse and build the model), rounded to full seconds. The final three columns provide the average number of actions for the original SG as well as all considered subgames $\mod{\game}$ in the main stage, and lastly the fraction of stages considered, i.e. the stages solved by the algorithm compared to the theoretically maximal possible number of stages ($2^n-1$).}
	\medskip
	%\caption{} 
	\centering
	\begin{tabular}{@{}lr c rrr c rcr c r@{}}
		\toprule
				&  		& \phantom{abc}	&	\multicolumn{3}{c}{Time}	& \phantom{ab} & \multicolumn{3}{c}{Avg. actions}	& \phantom{ab} & \\
				\cmidrule{4-6} \cmidrule{8-10}
		Model 	&  \multicolumn{1}{c}{$|S|$}&	& Lex.	&	\phantom{.}First	& \phantom{.} All	&& $\game$~~ &\phantom{a} &	$\mod{\game}$~~~~		&& Stages\\
		
		\midrule
		
		\textbf{R -- R} \\
		Dice[10] &  4,855 &&  $<$1 & $<$1 & $<$1 && 1.42  && 1.41 && 1/3 \\
		Dice[20] & 16,915 && $<$1 & $<$1 & $<$1 && 1.45 && 1.45 && 1/3 \\
		Dice[50] & 96,295 && 3 	 & 2     & 2    && 1.48 && 1.48 && 1/3 \\
		\textbf{S -- R} &&&&&&& \\
		Charlton & 502 	  && $<$1 & $<$1 & $<$1 && 1.56 && 1.07 && 3/3 \\
		\textbf{R -- S} &&&&&&& \\
		HW[5$\times$5] &   25,000 && 10 & 7.15 &   7 && 2.44 && 1.02 && 3/3  \\
		HW[8$\times$8] &  163,840 && 152 & 117 & 117 && 2.50 && 1.01 && 3/3  \\
		HW[10$\times$10]& 400,000 && 548 & 435 & 435 && 2.52 && 1.01 && 3/3  \\
		 \textbf{S--R--R} &&&&&&& \\
		AV[10$\times$10] &   106,524 && 15 & $<$1 & 10 && 2.17 && 1.55, 1.36 && 4/7\\
		AV[15$\times$15] &   480,464 && 85 & $<$1 & 50 && 2.14 && 1.52, 1.36 && 4/7\\
		AV[20$\times$20] & 1,436,404 && 281 & 3 & 172 && 2.13 && 1.51, 1.37 && 4/7\\
		\bottomrule
	\end{tabular}
	\label{tab:results}
\end{table}

%\begin{table}[htb]
%	\caption{Experimental Results. \textcolor{red}{M: will improve}}
%	\medskip
%	%\caption{} 
%	\centering
%	\scriptsize
%	\begin{tabular}{l  c  c | c  c  c | c  c  c}
%		%\hline
%		Model & Lex-obj. & $|S|$ & Time lex.  & \multicolumn{2}{c|}{Time single-obj.} & \multicolumn{2}{c}{Avg. \#actions} & \#Stages \\
%		\cline{5-8}
%		& type & & & first & all & ~~original~~ & sub-games  & considered \\
%		\hline
%		\hline
%		Dice[10] & R--R & 4,140 &  0.153s & 0.106s & 0.136s & 1.40  & 1.39 & 1/3 \\
%		Dice[20] & 	& 16,900 & 0.274s & 0.176s  & 0.222s & 1.45 & 1.45 & 1/3 \\
%		Dice[50] & 	& 96,300 & 2.51s & 1.55s  & 2.09s & 1.48 & 1.48 & 1/3 \\
%		\hline
%		Charlton & S--R & 502 &  0.022s & 0.028s  & 0.049s & 1.56 & 1.07 & 3/3 \\
%		\hline
%		AV[10$\times$10] & S--R--R & 107,000 &  15.2s & 0.261s & 9.64s & 2.17 & 2.10, 1.77 & 4/7\\
%		AV[15$\times$15] &  & 480,000 & 85.9s & 0.923s & 51.1s & 2.14 & 2.10, 1.78 & 4/7\\
%		AV[20$\times$20] &  & 1,440,000 & 281s & 2.85s & 172s & 2.13 & 2.10, 1.79 & 4/7\\
%		\hline
%		HW[5$\times$5] & R--S & 25,000 & 9.73s & 7.15s & 7.16s & 2.44 & 1.07 & 3/3  \\
%		HW[8$\times$8] &  & 164,000 & 152s & 117s & 117s & 2.50 & 1.01 & 3/3  \\
%		HW[10$\times$10] &  & 400,000 & 548s & 435s & 435s & 2.52 & 1.01 & 3/3  \\
%		%\hline
%	\end{tabular}
%	\label{tab:results2}
%\end{table}

The case studies are grouped by the type of lex-objective, where R indicates reachability, S safety.
For each combination of case study and scaling parameters, we report the state size in column $|S|$, three different model checking runtimes, the average number of actions in the original and all considered restricted games, and the fraction of stages considered, i.e. the stages solved by the algorithm compared to the theoretically maximal possible number of stages ($2^n-1$).

We compare the time of our algorithm on the lexicographic objective (Lex.) to the time for checking the first single objective (First) and the sum of checking all single objectives (All). We see that the runtimes of our algorithm and checking all single objectives are always in the same order of magnitude. This shows that our algorithm works well in practice and that the overhead is often small. 
Even on SGs of non-trivial size (HW[10$\times$10] and AV[20$\times$20]), our algorithm returns the result within a few minutes.

Regarding the average number of actions, we see that the decrease in the number of actions in the sub-games $\mod{\game}$ obtained by restricting the input game to optimal actions varies:  
For example, very few actions are removed in the Dice instances, in AV we have a moderate decrease and in HW a significant decrease, almost eliminating all non-determinism after the first objective. 
It is our intuition that the less actions are removed, the higher is the overhead compared to the individual single-objective solutions.
Consider the AV and HW examples: While for AV[20$\times$20], computing the lexicographic solution takes 1.7 times as long as all the single-objective solutions, it took only about 25\% longer for HW[10$\times$10]; this could be because in HW, after the first objective only little nondeterminism remains, while in AV also for the second and third objectives lots of choices have to be considered.
Note that the first objective sometimes (HW), but not always (AV) needs the majority of the runtime.  

We also see that the algorithm does not have to explore all possible stages. For example, for Dice we always just need a single stage, because the SG is absorbing. For charlton and HW all stages are relevant for the lex-objective, while for AV 4 of 7 need to be considered. 
\section{Conclusion and Future Work}
\label{sec:conc}

In this work we considered simple stochastic games with lexicographic
reachability and safety objectives.
Simple stochastic games are a standard model in reactive synthesis of
stochastic systems, and lexicographic objectives let one consider
multiple objectives with an order of preference.
We focused on the most basic objectives: safety and reachability.
While simple stochastic games with lexicographic objectives have
not been studied before, we have presented (a)~determinacy; (b)~strategy
complexity; (c)~computational complexity; and (d)~algorithms;
for these games.
Moreover, we showed how these games can model many different case studies
and we present experimental results for them.

There are several directions for future work.
First, for the general case closing the complexity gap ($\nexptime \cap \conexptime$ upper
bound and $\pspace$ lower bound) is an open question.
Second, the study of lexicographic
simple stochastic games with more general objectives, e.g., quantitative 
or parity objectives poses interesting questions. In particular, in the case of parity objectives, there are some indications that the problem is significantly harder:
Consider the case of a reachability-safety lex-objective. If the lex-value is $(1,1)$ then both objectives can 
be guaranteed almost surely. Since almost-sure safety is sure safety, our 
results imply that sure safety and almost-sure reachability can be 
achieved with constant memory. In contrast, for parity objectives the 
combination of sure and almost-sure requires infinite-memory (e.g, see~\cite[Appendix A.1]{DBLP:journals/corr/abs-1804-03453}). 
%Moreover, for parity objectives in stochastic games, even qualitative combinations for more than two  objectives remain open~\cite{DBLP:journals/corr/abs-1804-03453}. 
%
% ---- Bibliography ----
%
% BibTeX users should specify bibliography style 'splncs04'.
% References will then be sorted and formatted in the correct style.
%
\bibliographystyle{splncs04}
\bibliography{references}
%

% ---- Appendix ----
\ifarxivelse{}{
\appendix % Tells latex that appendix starts here
\section{Appendix -- Full Proofs}

\subsection{Proof of Lemma \ref{lem:fact_about_strats} (Two facts about lex-optimal actions)}
\label{app:facts}

\begin{enumerate}[label=(\alph*)]
	\item
	Recall that if $\maxstrat$ is a lex-optimal MD strategy and $\minstrat$ a lex-optimal MD counter-strategy, then $\prob_s^\bothstrats(\objvec) = \lexval(s)$ for all $s \in S$.
	
	Now let $s \in \Splmax$ (if $\Splmax = \emptyset$ then there is nothing to show). Suppose that $\maxstrat(s) = a \in \Act(s)$ where $a$ is an action that is \emph{not} lex-optimal. Then
	\begin{align*}
		\prob_s^\bothstrats(\objvec) &= \sum_{s'}P(s,a,s')\prob_{s'}^\bothstrats(\objvec) & \text{ (because $\bothstrats$ are memoryless)}\\
		&= \sum_{s'}P(s,a,s') \lexval(s') & \text{ (because $\bothstrats$ are lex-optimal)}\\
		& \lesslex \lexval(s) &\text{ (because action $a$ is not lex-optimal)}
	\end{align*}
	which is a contradiction. The case $s \in \Splmin$ is analogous.
	%\item Really not surprising. Notice that this would be a consequence of the first item if we already knew that MD optimal strategies exist. To formally prove this one can plug in the sup inf definition of the lex-value and check that it holds (see thesis). Notice that this only holds because $\objvec$ is absorbing.
	\item %Also not surprising but probably the most interesting item. Still not a 100 percent sure how to prove this in a simpe existence argument to the single objective case. I should definitely write a formal proof for this (at least for the appendix). 
	%If the statement was false, then clearly $\mod{\mathbf{v}}^\lexabbr(\sinit) \lesslex \lexval(\sinit)$ for some $\sinit \in S$.
	
	Let us extend the notion of lex-value to finite paths, that is we define for $\path s \in (S \times \actlabels)^* \times S$ the value
	\begin{equation*}
	\lexval(\path s) = \adjustlimits \sup_{\maxstrat} \inf_{\minstrat} \prob_{\path s}^\bothstrats(\objvec)
	\end{equation*}
	where $\prob_{\path s}$ is the probability measure in the induced MC $\game^\bothstrats$ with starting state $\path s$. We have
	\begin{equation*}
	\label{eq:remove_path}
	\adjustlimits\sup_{\maxstrat} \inf_{\minstrat} \prob_{\path s}^\bothstrats(\objvec)
	= \adjustlimits\sup_{\maxstrat(\path)} \inf_{\minstrat(\path)} \prob_{s}^{\maxstrat(\path),\minstrat(\path)}(\objvec)
	= \adjustlimits\sup_{\maxstrat} \inf_{\minstrat} \prob_{s}^\bothstrats(\objvec)
	= \lexval(s)
	\end{equation*}
	where $\maxstrat(\path)$ is the strategy that behaves like $\maxstrat$ after seeing path $\path$.
	
	The lex-values for paths satisfy the following equations:
	\begin{align*}
			& &\lexval(\path s) = \max_{a \in \Act(s)} \sum_{s'}P(s,a,s')\lexval(\path s a s') & \hspace*{5mm}\text{ if } s \in \Splmax \\
			&\text{ and }&\lexval(\path s) = \min_{a \in \Act(s)} \sum_{s'}P(s,a,s')\lexval(\path s a s') & \hspace*{5mm}\text{ if } s \in \Splmin.
	\end{align*}
	Notice that the equations trivially hold for all paths that already reached an absorbing state because all actions available at a sink are lex-optimal.
	
	Now let $\maxstrat \in \maxstrats$ be a strategy of $\plmax$ that selects a lex-suboptimal action with positive probability after seeing a finite path $\path s$, $s \in \Splmax$. Then:
	\begin{align*}
		\inf_\minstrat \prob_{\path s}^\bothstrats(\objvec) =& \sum_{a \in \Act(s)} \maxstrat(a) \sum_{s'}P(s,a,s') \inf_\minstrat \prob_{\path s a s'}^\bothstrats(\objvec) &\\
		\leqlex& \sum_{a \in \Act(s)} \maxstrat(a) \sum_{s'}P(s,a,s') \lexval(\path s a s') &\\
		=& \sum_{a \in \Act(s)} \maxstrat(a) \sum_{s'}P(s,a,s') \lexval(s') &\text{(by \eqref{eq:remove_path})}\\
		\lesslex& \lexval(\path s). & \aligncomment{because $a$ is lex-suboptimal}
	\end{align*}
	Thus if $\maxstrat$ had played a lex-optimal action after seeing path $\path s$ instead, it would have achieved a strictly greater lex-value. Thus lex-suboptimal actions do not play a role for player $\plmax$ if the lex-value should be maximized. Hence one can remove all those actions without changing the lex-value. The argument for $\plmin$ is similar.

	\qed
\end{enumerate}

\subsection{Proof of Lemma \ref{lem:THE_lemma} \\(Characterizing lex-optimal MD strategies via the final set $F$)}
\label{app:theLemma}

We first prove the following about the special case of MDPs:

\begin{lemma}
	\label{lem:mdp_md}
	Let $\game$ be an MDP (i.e., $\Splmax = \emptyset$ or $\Splmin = \emptyset$) and let $\objvec$ be an absorbing lex-objective. Then there exists an MD lex-optimal strategy for $\objvector$ for the respective player.
\end{lemma}
\begin{proof}
	Let us assume that $\Splmin = \emptyset$, this is no loss of generality as otherwise we can exchange all $\reach{S_i}$ for $\safe{S_i}$ in $\objvec$ and swap the roles of $\plmin$ and $\plmax$. Fix a state $s \in S$. It is known that the set of points $\vec{x} \in [0,1]^n$ such that there exists a strategy $\sigma \in \maxstrats$ with 
	\[	
	\left( \prob_s^\sigma(\obj_1),\ldots,\prob_s^\sigma(\obj_n)\right) \dot{\geq}\ \vec{x}
	\]
	where $\dot{\geq}$ denotes point-wise inequality is a \emph{closed convex polyhedron} $\mathfrak{P}$ \cite{FKP12}, \cite{DBLP:journals/lmcs/EtessamiKVY08} which is contained in $[0,1]^n$. Therefore $\mathfrak{P}$ contains a maximum $\vec{x}^*$ in the order $\leqlex$. Moreover, $\vec{x}^*$ is a vertex of $\mathfrak{P}$, i.e., a point contained in $\mathfrak{P}$ which is not a proper convex combination of two \textit{different} points of $\mathfrak{P}$. If not, then $\vec{x}^* = \alpha \vec{y} + (1-\alpha)\vec{z}$ for $\vec{y} \neq \vec{z} \in \mathfrak{P}$ and $0 < \alpha < 1$. Let $i$ the tiebraker position of $\vec{y}$ and $\vec{z}$. We can assume w.l.o.g. that $\vec{y}_i > \vec{z}_i$. But then it follows immediately that $\vec{y} \grlex \vec{x}^*$, contradiction to the fact that $\vec{x}^*$ was maximal in $\mathfrak{P}$. The claim follows because the vertices of $\mathfrak{P}$ are achieved by MD strategies \cite{FKP12}.
	\qed
\end{proof}

%Thus, for MDPs in the absorbing case, lex-optimal strategies are not more complex than optimal single-objective reachability/safety strategies \cite{Puterman}.
%In the rest of this section we prove that the same is true for games. The next lemma is our main structural result about lex-optimal strategies in games.\todo{put somewhere?}
For proving Lemma \ref{lem:THE_lemma}, we need to show the following intermediate result:

\begin{lemma}
	\label{lem:THE_lemma_helper}
	Let $\maxstrat \in \maxstratsmd$ and let $\objvec$ be an absorbing lex-objective. Then $\maxstrat$ is lex-optimal for $\objvector$ if and only if
	$\maxstrat$ is locally lex-optimal and for all $1 \leq i \leq n$ such that $\obj_i = \reach{S_i}$ and all $s \in S$ it holds that
	\begin{equation}
	\label{eq:THE_lemma_helper}
	\forall \minstrat \in \minstratsmd\colon \prob_s^{\maxstrat,\minstrat}(\reach{S_i \cup \valzeroset_i}) = 1. \tag{$\triangle$}
	\end{equation}
\end{lemma}

\begin{proof} We show the two directions of the ``if and only if'' statement. Recall that an MC can be simplified to a tuple $\mc = (S,P)$ such that $P: S \rightarrow \dist(S)$.
	
	\refif: 
	We use the following characterization of the reachabiliy probabilities in any (not necessarily finite) Markov chain: The  probabilities $\prob_s(\reach{S'})$ constitute the least fixpoint $x(s)$ of the operator
	\begin{equation}
	\label{eq:reachop}
	\reachop \colon [0,1]^S \rightarrow [0,1]^S,\ \reachop(x)(s) =
	\begin{cases}
	1 &\text{if } s \in S\\
	\sum_{s'} P(s,s') x(s') &\text{else}
	\end{cases}
	\end{equation}
	which is monotonic on the complete lattice $[0,1]^S$ (that is, the set of all mapping from $S$ to $[0,1]$) \cite{BK08}. In a finite MC, the fixpoint of $\reachop$ can be made unique be requiring additionally that $\reachop(x)(s) = 0$ if there is not path from $s$ to $S'$ in the MC.
	
	We now prove the \refif-direction by induction on $n$. We first show the inductive step and then argue that the base case $n=1$ follows with a similar, slightly simpler argument. Thus let $n > 1$. Moreover, let $\maxstrat \in \maxstratsmd$ be locally lex-optimal and assume that \eqref{eq:THE_lemma_helper} holds. To prove that $\maxstrat$ is lex-optimal, we let $\minstrat \in \minstrats$ be a lex-optimal-counter strategy against $\maxstrat$ and show that $\prob_s^{\maxstrat,\minstrat}(\obj_i) = \lexval_i(s)$ for all $1 \leq i \leq n$. By the previous Lemma \ref{lem:mdp_md}, we can assume that $\minstrat$ is MD.
	By the I.H., $\maxstrat$ is already lex-optimal for $\objvec_{<n} = (\obj_1,\ldots,\obj_{n-1})$.
	Next observe that since $\minstrat$ is a lex-optimal counter-strategy against $\maxstrat$, it holds that
	\begin{equation}
	\label{eq:easy_inequality}
	\prob_s^{\maxstrat,\minstrat}(\obj_n) \leq \lexval_n(s).
	\end{equation}
	Thus we only need to prove the other inequality ``$\geq$'' in \eqref{eq:easy_inequality}. Therefore we make a case distinction according to type of $\obj_n$:
	\begin{itemize}
		\item $\obj_n = \safe{S_n}$.
		Consider the MC $\game^\bothstrats$. Since $\bothstrats$ are both MD, this MC has the same finite state space $S$ as the game and its transition probability function is defined as $P^\bothstrats(s,s') = P(s,\maxstrat(s),s')$ if $s \in \Splmax$ and $P^\bothstrats(s,s') = P(s,\minstrat(s),s')$ if $s \in \Splmin$. In $\game^\bothstrats$, the safety probabilities $\prob_s(\obj_n) = \prob_s(\safe{S_n})$ constitute the greatest fixpoint of the operator
		\[
		\safeop \colon [0,1]^S \rightarrow [0,1]^S,\ \safeop(x)(s) =
		\begin{cases}
		0 &\text{if } s \in S_n\\
		\sum_{s'} P^\bothstrats(s,s') x(s') &\text{else}
		\end{cases}
		\]
		which is obtained from the operator $\reachop$ for reachability using the relation $\prob_s(\safe{S_n}) = 1 - \prob_s(\reach{S_n})$.
		Just like $\reachop$, the operator $\safeop$ is also monotonic on the complete lattice $[0,1]^S$ and we can apply the well-known Theorem of Knaster \& Tarski: If we can prove that for all $s \in S$
		\begin{equation}
		\label{eq:the_lemma_safety_proof_obl}
		\lexval_n(s) \leq \safeop(\lexval_n)(s)
		\end{equation}
		then this implies $\lexval_n(s) \leq (\gfp \safeop)(s) = \prob_s^{\maxstrat,\minstrat}(\obj_n)$, where $\gfp \safeop$ denotes the greatest fixpoint of $\safeop$. To prove \eqref{eq:the_lemma_safety_proof_obl}, we let $s \in S$ and make another case distinction:
		\begin{itemize}
			\item $s \in S_n$. In this case clearly $\lexval_n(s) = 0 \leq \safeop(\lexval_n)(s)$.
			\item $s \in \Splmax \setminus S_n$. Then 
			\begin{align*}
			\lexval(s) &= \max_{a \in \Act(s)} &&\hspace*{-12.5mm}\sum_{s'} P(s,a,s')\lexval(s')
			\aligncomment{by Lemma \ref{lem:fact_about_strats} \ref{facts:c}}\\
			&= &&\hspace*{-12.5mm}\sum_{s'} P(s,\maxstrat(s),s')\lexval(s')
			\aligncomment{because $\maxstrat$ is locally lex-optimal}
			\end{align*}
			and thus in particular $\lexval_n(s) = \safeop(\lexval_n)(s)$.
			\item $s \in \Splmin \setminus S_n$. Let $\lexval_{<n}(s)$ be the lex-value with respect to the first $n-1$ objectives $\objvec_{<n}$. Then we have since $\maxstrat$ is lex-optimal for $\objvec_{<n}$ and $\minstrat$ is a lex-optimal counter-strategy against $\maxstrat$, that
			\begin{align*}
			\lexval_{<n}(s) &= \min_{a \in \Act(s)} &&\sum_{s'} P(s,a,s')\lexval_{<n}(s') \aligncomment{by Lemma \ref{lem:fact_about_strats} \ref{facts:a}}\\
			&= &&\sum_{s'} P(s,\minstrat(s),s')\lexval_{<n}(s')
			\end{align*}
			Let $\Act_{<n}(s)$ be the lex-optimal actions available in $s$ with respect to $\objvec_{<n}$. By the previous equation, $\minstrat(s) \in \Act_{<n}(s)$. Therefore,
			\begin{align*}
			\lexval_n(s) &= \min_{a \in \Act_{<n}(s)} &&\sum_{s'} P(s,a,s')\lexval_n(s')\\
			&\leq &&\sum_{s'} P(s,\minstrat(s),s')\lexval_n(s') = \safeop(\lexval_n)(s).
			\end{align*}
		\end{itemize}
		Thus we have $\lexval_n(s) = \prob_s^{\maxstrat,\minstrat}(\obj_n)$ together with \eqref{eq:easy_inequality} and $\maxstrat$ is lex-optimal for $\objvector = (\obj_1,\ldots,\obj_n)$.
		\item $\obj_n = \reach{S_n}$. This case is proved in a similar though slightly more complicated way than the previous case. As mentioned earlier, in $\game^{\maxstrat,\minstrat}$ the probabilities $\prob_s(\obj_n)$ constitute the \emph{unique} fixpoint of the following monotonic operator:
		\[
		\reachop \colon [0,1]^S \rightarrow [0,1]^S,\ \reachop(x)(s) =
		\begin{cases}
		1 &\text{if } s \in S_n\\
		0 &\text{if } s \text{ cannot reach } S_n \\
		\sum_{s'} P(s,s') x(s') &\text{else}
		\end{cases}
		\]
		where the transition probability function $P$ of the Markov chain is defined as before.	As in the other case, we prove that $\lexval_n(s) \leq \reachop(\lexval_n)(s)$ for all $s \in S$, which implies $\lexval_n(s) \leq (\gfp \reachop)(s) = \prob_s^{\maxstrat,\minstrat}(\obj_n)$. Notice that the greatest fixpoint $\gfp \reachop$ is equal to the unique fixpoint of $\reachop$. Let $s \in S$ and let us again make a case distinction to prove $\lexval_n(s) \leq \reachop(\lexval_n)(s)$ for all $s$:
		\begin{itemize}
			\item If $s \in S_n$, then $\lexval_n(s) = 1 = \reachop(\lexval_n)(s)$.
			\item The cases where $s$ can reach $S_n$ but $s \notin S_n$ can be shown exactly as in the previous case where $\obj_n$ was safety.
			%\item $s \in \Splmax \setminus S_n$ and $s$ \emph{can} reach $S_n$ in $\game^{\maxstrat,\minstrat}$. Can be shown as in the previous case.
			%\item $s \in \Splmin \setminus S_n$ and $s$ \emph{can} reach $S_n$ in $\game^{\maxstrat,\minstrat}$. Can also be shown as in the previous case.
			\item Now suppose $s$ \emph{cannot} reach $S_n$ in $\game^{\maxstrat,\minstrat}$, i.e. $\prob_s^{\maxstrat,\minstrat}(\reach{S_n}) = 0$. In this case we need to show that $\lexval_n(s) = 0$, or equivalently, $s \in \valzeroset_n$.
			
			By condition \eqref{eq:THE_lemma_helper}, we have for all $t \in S$ that
			\begin{align*}
			&\prob_t^\bothstrats(\reach{S_n \cup \valzeroset_n }) \\
			= &\prob_t^\bothstrats(\reach{S_n}) + \prob_t^\bothstrats(\reach{\valzeroset_n})\\
			= &\prob_t^\bothstrats(\reach{S_n}) + 1 - \prob_t^\bothstrats(\safe{\valzeroset_n})\\
			= &1
			\end{align*}
			and thus $\prob_t^\bothstrats(\reach{S_n}) = \prob_t^\bothstrats(\safe{\valzeroset_n})$. Therefore, $\maxstrat$ is also locally lex-optimal for the objective $(\obj_1,\obj_2,\ldots,\safe{\valzeroset_n})$. But then we can show exactly as in the previous case that $\lexval_n(t) \leq \safeop(\lexval_n)(t)$ where $\safeop$ is the fixpoint operator for safety probabilities associated to the objective $\safe{\valzeroset_n}$. This implies that $\lexval_n(s) \leq (\gfp \safeop)(s) = \prob_s^{\maxstrat,\minstrat}(\safe{\valzeroset_n}) = 0$.
			
		\end{itemize}
	\end{itemize}
	Finally, for the base case $n=1$ observe that the same reasoning applies with the simplification that we do not need to care about previous targets. In particular, we do not need to apply the I.H.
	
	\refonlyif: Let $\maxstrat \in \maxstratsmd$ be lex-optimal. First observe that $\maxstrat$ is also locally lex-optimal by Lemma \ref{lem:fact_about_strats} \ref{facts:a}.
	%(If that was not the case, let $s$ be a state such that $\maxstrat(s)$ is a sub-optimal action. Then clearly $\maxstrat$ will not achieve $\lexval(s)$ against an optimal counter-strategy of $\plmin$ when the game starts in $s$.)	
	Now let $i$ be such that $\obj_i = \reach{S_i}$, let $s \in S$ be any state and let $\minstrat \in \minstratsmd$. It remains to show \eqref{eq:THE_lemma_helper}. Assume for contradiction that $\prob_s^{\maxstrat,\minstrat}(\reach{S_i \cup \valzeroset_i}) < 1$. This means that in the \emph{finite} Markov chain $\game^{\maxstrat,\minstrat}$, there exists a bottom strongly connected component (BSCC) $B \subseteq S$ such that $B \cap (S_i \cup \valzeroset_i) = \emptyset$. Thus if $t \in B$ is a state, we have $\lexval_i(t) > 0$. Further it holds that $\prob_{t}^{\maxstrat,\minstrat}(\reach{S_i}) =0$ because $s$ can only reach states inside $B$, but $B \cap S_i = \emptyset$. This however is a contradiction to the lex-optimality of $\maxstrat$.
	\qed
\end{proof}

We can now finish the proof of Lemma \ref{lem:THE_lemma}:

\begin{proof}(of Lemma \ref{lem:THE_lemma})
	Let $\maxstrat$ be locally lex-optimal, let $s \in S$ and let $\minstrat \in \minstratsmd$. We show the following equivalence:
	\begin{equation*}
	\prob_s^\bothstrats\left(\reach{F} \right) = 1 \iff \forall i \in \indicesreach \colon \prob_s^\bothstrats(\reach{S_i \cup \valzeroset_i}) = 1,
	\end{equation*}
	where $R = \{i \leq n \mid \obj_i = \reach{S_i}\}$. The equivalence states that conditions \eqref{eq:THE_lemma_helper} and \eqref{eq:THE_lemma} are equivalent and thus Lemma \ref{lem:THE_lemma_helper} is equivalent to Lemma \ref{lem:THE_lemma}. For $R = \emptyset$ there is nothing to show, so we let $R \neq \emptyset$.
	
	To show direction ``$\Rightarrow$'', assume for contradiction that the left hand side holds but $\prob_s^\bothstrats(\reach{S_i \cup \valzeroset_i}) < 1$ for some $i \in R$. Then in the finite MC $\game^\bothstrats$ there exists a BSCC $B$ which is reachable from $s$ with positive probability and $B \cap (S_i \cup \valzeroset_i) = \emptyset$. Thus if $t \in B$, then $t \notin S_i$ and $t \notin \valzeroset_i$. Thus $t \notin F$, contradiction because $t$ is reachable from $s$ with positive probability.
	
	For direction ``$\Leftarrow$'', the argument is similar. Suppose that the right hand side holds but $\prob_s^\bothstrats(\reach{F}) < 1$. Then in the finite MC $\game^\bothstrats$ there exists a BSCC $B$ which is reachable from $s$ with positive probability and $B \cap F = \emptyset$. Let $t \in B$. Then since $t \notin F$, we have by definition that $t \notin S_i$ for all $i \in R$ and $t \notin \valzeroset_j$ for some $j \in R$. But this is a contradiction to $\prob_s^\bothstrats(\reach{S_j \cup \valzeroset_j})$ because $t$ is reachable from $s$ with positive probability.
	\qed
\end{proof}

\subsection{Proof of Theorem \ref{thm:md_exists}\\(Lex-optimal MD strategies exist for absorbing objectives)}
\label{app:thmMDexists}
	
Let $\mod{\game}$ be the game obtained by removing lex-sub-optimal actions for both players. Let $\val(s)$ be the value of state $s \in S$ for the objective $\reach{F}$ in the modified game $\mod{\game}$, where $F$ is the final set like in Lemma \ref{lem:THE_lemma} (we can assume that $\indicesreach \neq \emptyset$). We show that $\val(s) = 1$ for all $s \in S$. Assume towards contradiction that there exists a state $s$ with $\val(s) < 1$.
\begin{itemize}
	\item If $s \in \sinks(\game)$, then either $s \in S_i$ for some $i \in \indicesreach$, or otherwise $s$ is a sink which is not contained in any of the $S_i$ with $i \in R$ and thus $s \in \valzeroset_i$ for all $i \in R$. Thus $s \in F$ by definition of $F$ and $\val(s) = 1$, contradiction.
	\item Let $s \notin \sinks(\game)$. Let $\maxstrat$ be an MD optimal strategy for $\reach{F}$ in $\mod{\game}$ and let $\minstrat$ be an MD optimal counter-strategy. Notice that such strategies exist because we are only considering a single objective \cite{Con92}. As usual, we consider the finite MC $\mod{\game}^\bothstrats$. Since $\val(s) < 1$, we have $\prob^\bothstrats_s(\reach{F}) < 1$ which means that there is a BSCC $B \subseteq S$ in $\mod{\game}^\bothstrats$ such that $B \cap F = \emptyset$ and $\prob^\bothstrats_s(\reach{B}) > 0$. Let $t \in B$ be any state in the BSCC. Then clearly, $\prob^\bothstrats_t(\reach{F}) = 0$ and thus $\val(t) = 0$ because $\maxstrat$ is optimal for $\reach{F}$. But since $t \notin F$, we have by definition of $F$ that $\exists i \in R \colon t \notin \valzeroset_i$, which means that $\lexval_i(t) > 0$. Notice that here, $\lexval$ are the lex-values in the \emph{original} game $\game$, however by Lemma \ref{facts:c}, they coincide with the lex-values in $\mod{\game}$. Thus since $\lexval_i(t) > 0$, there is a strategy of $\plmax$ in $\mod{\game}$ that reaches $S_i$ with positive probability against all counter-strategies of $\plmin$ and thus also reaches $F$ with positive probability because $S_i \subseteq F$. This is a contradiction to $\val(t) = 0$.
\end{itemize}	
\qed

\subsection{Proof of Lemma \ref{lem:safety_reduction} (Reduction Safe $\rightarrow$ Reach)}
\label{app:safetyRed}

Let $\maxstrat \in \maxstratsmd$ be lex-optimal for $\objvec$ (such a $\maxstrat$ exists by Theorem \ref{thm:md_exists}).
Clearly, $\maxstrat$ is in particular lex-optimal for the first $n-1$ objectives $\objvec_{<n}$.
Let us denote by $\maxstrats^{<n}$ the set of all MD lex-optimal strategies for player $\plmax$ with respect to $\objvec_{<n}$. We have $\maxstrat \in \maxstrats^{<n}$.

We know by Lemma \ref{lem:THE_lemma} that $\maxstrat$ reaches $F_{<n} = F$ almost-surely against all $\minstrat \in \minstratsmd$.
Now fix an optimal-counter strategy $\minstrat \in \minstrats^{<n}(\maxstrat)$, the set of all MD lex-optimal counter strategies against $\maxstrat$. Then for all $s \in S$, it holds that
\begin{align*}
&&\prob_s^\bothstrats(\safe{S_n}) &= \sum_{t \in F} \prob_s^\bothstrats(\until{(S \setminus F)}{t}) \cdot \prob_{t}^\bothstrats(\safe{S_n})
\aligncomment{because F is reached a.s. and $(S\setminus F) \cap S_n =\emptyset$ as $\objvec$ is absorbing}
\implies&&\lexval_n(s) &=  \sum_{t \in F} \prob_s^\bothstrats(\until{(S \setminus F)}{t}) \cdot \lexval_n(t) \aligncomment{because $\bothstrats$ are lex-optimal}\\
&&&= \adjustlimits \sup_{\maxstrat' \in \maxstrats^{<n}}\inf_{\minstrat' \in \minstrats^{<n}(\maxstrat')} \sum_{t \in F} \prob_s^{\maxstrat',\minstrat'}(\until{(S \setminus F)}{t}) \cdot \lexval_{n}(t)\\
\aligncomment{because $\maxstrat$ is lex-optimal for $\objvec_{<n}$}
&&&=\adjustlimits \sup_{\maxstrat' \in \maxstrats^{<n}}\inf_{\minstrat' \in \minstrats^{<n}(\maxstrat')}\prob_s^{\maxstrat',\minstrat'}(\reach{q_n})\\
\aligncomment{by definition}\\
&&&=\ ^{\objvec'}\lexval_{n}(s).
\end{align*}
This proves the claim.
\qed

\subsection{Proof of Theorem \ref{thm:alg_correct} (Algorithm $\mathtt{SolveAbsorbing}$ is correct)}
\label{app:algAbsCorr}

The proof is by induction on $n$. For $n=1$, the algorithm is correct by the assumption that $\algsingleobj$ is correct in the single-objective case. 
Next we show the inductive step $n>1$ for reachability and then for safety objectives.
%This case distinction corresponds to the if-statements in lines ? and ?:
\begin{itemize}
	
	\item \textit{Case 1:} $\obj_n = \reach{S_n}$. By the I.H., $\mod{\game}$ is the correct restriction of $\game$ to lex-optimal actions for both players with respect to the first $n-1$ objectives $\objvector_{<n} = (\obj_1,\ldots,\obj_{n-1})$ and $\maxstrat$ is a lex-optimal MD strategy in $\game$ with respect to $\objvector_{<n}$. The algorithm then correctly computes $\mod{\maxstrat}$, an MD optimal strategy in $\mod{\game}$ with respect to the single-objective $\reach{S_n}$, and the single-objective values $\val(s)$ of this objective in $\mod{\game}$ by calling $\algsingleobj$ (line \ref{line:startSingleObj}).
	The strategy $\maxstrat \in \maxstratsmd$ is then updated as follows:
	\[
	\maxstrat(s) =
	\begin{cases}
	\mod{\maxstrat}(s) &\text{if } \mod{\val}(s) > 0\\
	\maxstrat_{old}(s) &\text{if } \mod{\val}(s) = 0.
	\end{cases}
	\] 
	We claim that $\maxstrat$ is lex-optimal for the whole lex-objective $\objvector = (\obj_1,\ldots,\obj_{n})$. 
	\begin{itemize}
		
		\item We first show that $\maxstrat$ remains lex-optimal for the first $n-1$ objectives $\objvector_{<n}$ by applying the \refif-direction of Lemma \ref{lem:THE_lemma_helper}: First observe that by definition, $\maxstrat$ is locally lex-optimal with respect to $\objvector_{<n}$. Therefore it only remains to show condition \eqref{eq:THE_lemma_helper} in Lemma \ref{lem:THE_lemma_helper}. Let $i<n$ such that $\obj_i = \reach{S_i}$, let $s \in S$ and let $\minstrat \in \minstratsmd$ be a counter-strategy against $\maxstrat$. If $\val(s) = 0$, then $\prob_s^\bothstrats(\reach{S_i \cup \valzeroset_i}) = 1$ because from initial state $s$, $\maxstrat$ behaves like $\maxstrat_{old}$ which is lex-optimal for $\objvector_{<n}$. Thus let $\val(s) > 0$. From the \refonlyif-direction of Lemma \ref{lem:THE_lemma_helper} applied to $\mod{\maxstrat}$, we know that $\prob^{\maxstrat,\minstrat}_s(\reach{S_n \cup \{s \in S \mid \val(s) = 0\}}) = 1$. Thus for a play $\infpath$ that starts in $s$ and is consistent with $\bothstrats$, almost-surely one of the following two cases occurs:
		\begin{itemize}
			\item If $\infpath$ reaches a state $t \in S_n$, then since $t$ is a sink, we have $t \in S_i \cup \valzeroset_i$.
			\item If $t$ with $\val(t) = 0$ is reached in $\infpath$, then since we play according to $\maxstrat_{old}$ from $t$, we either reach $S_i$ or $\valzeroset_i$ by the \refonlyif-direction of Lemma \ref{lem:THE_lemma} applied to $\maxstrat_{old}$.
		\end{itemize}
		Thus $\prob_s^\bothstrats(\reach{S_i \cup \valzeroset_i}) = 1$ and $\maxstrat$ remains lex-optimal for $\objvector_{<n}$.
		
		\item To complete the proof that $\maxstrat$ is lex-optimal for $\objvector$, notice that $\val(s) = \lexval_n(s)$ for all $s \in S$ by Lemma \ref{facts:c}.
		%First notice that $\mod{\val}(s) \leq \lexval_n(s)$ holds because $\maxstrat$ achieves the values $\mod{\val}(s)$ and is already lex-optimal for $\objvec_{<n}$. The inequality $\mod{\val}(s) \geq \lexval_n(s)$ follows because \todo{T: finish this}
		
	\end{itemize}
	
	\item \textit{Case 2:} $\obj_n = \safe{S_n}$. 
	Since by the I.H., the values $\lexval_{1},\ldots,\lexval_{n-1}$ are the correct lex-values with respect to $\objvec_{<n}$, the algorithm computes the correct final set $F_{<n} = F_{<n+1} = F$.
	%Notice that $F$ is the same for both $\objvec_{<n}$ and $\objvec$ because $\obj_n$ is safety.
	
	Next observe that $\val(s) = \lexval_n(s)$ for all $s \in F$ because of the following:
	\begin{itemize}
		\item First, $\mod{\maxstrat}$ is locally lex-optimal w.r.t. $\objvec_{<n}$ because it is defined in the subgame $\mod{\game}$. Therefore by Lemma \ref{lem:THE_lemma}, $\maxstrat$ is already (globally) lex-optimal for $\objvec_{<n}$ from all $s \in F$ because condition \ref{eq:THE_lemma} is satisfied trivially. Thus $\val(s) \leq \lexval_n(s)$.
		\item Second, by the same argument, an MD lex-optimal strategy for $\objvec$ is necessarily locally lex-optimal w.r.t. $\objvec_{<n}$. The strategy $\mod{\maxstrat}$ is locally lex-optimal w.r.t. $\objvec_{<n}$ and moreover optimal for $\obj_n$ in the subgame $\mod{\game}$. Thus ${\val}(s) \geq \lexval_n(s)$.
	\end{itemize}		 
	
	Notice that it is very well possible that ${\val}(s) > \lexval_n(s)$ for $s \notin F$ because the strategy $\mod{\maxstrat}$ does not necessarily reach $F$ from $s \notin F$.
	
	Applying Lemma \ref{lem:safety_reduction} concludes the proof: The quantified reachability objective $\quanfun_n$ constructed by the algorithm indeed satisfies $\quanfun_n(s) = {\val}(s) = \lexval_n(s)$ for all $s \in F$ as we have just shown. The result strategy $\maxstrat$ defined by the algorithm is
	\[
	\maxstrat(s) =
	\begin{cases}
	\mod{\maxstrat}(s) &\text{if } s \in F\\
	\maxstrat_{Q}(s) &\text{if } s\notin F
	\end{cases}
	\]
	where $\maxstrat_{Q}$ is a lex-optimal strategy for $\objvec'=(\obj_1,\ldots,\obj_{n-1},\reach{\quanfun_n})$. Thus with Lemma \ref{lem:safety_reduction} and the above discussion, $\maxstrat$ is lex-optimal from \emph{all} states.
	
\end{itemize}
\qed

\subsection{Proof of Lemma \ref{lem:non_absorbing_eq} (Reduction General $\rightarrow$ Absorbing)}
\label{app:redGenAbsorb}

Note that both $^{\objvec}\lexval$ and $^{\qobjvec}\lexval$ depend on the same SG $\game$.
We again proceed by an induction on the number of targets. 
The single objective case trivially holds, because then $q_1$ is exactly the function that is $1$ for all $s \in S_1$, and the objective is correctly set to reachability or safety.

The induction hypothesis states that for lex-objectives $\objvec$ of length at most $n-1$, we have $^{\objvec}\lexval =\ ^{\qobjvec}\lexval$.

For the induction step, let $n > 1$ and consider $\objvec = (\obj_1, \dots, \obj_{n})$ and an arbitrary state $s$. In this proof, we write $\alltargets = \bigcup_{j\leq n} S_j$ for the sake of readability.
\begin{itemize}
	\item If $s \in \alltargets$, then for all $1 \leq i \leq n$, $\quanfun_i(s) =$ $ ^{\objvec}\lexval_i(s) =$ $^{\qobjvec}\lexval_i(s)$. This is the case, because if $s \in S_i$, then $\quanfun_i(s) = 1$, which is correct.
	Else if $s \notin S_i$, then $\quanfun_i(s) =\ ^{\objvec(s)}\lexval_{i}(s)$, where $\objvec(s)$ has less than $n+1$ objectives, because $s \in S_j$ for some $j$. Thus, by the induction hypothesis, $^{\objvec(s)}\lexval_{i}(s)$ $=$ $^{\qobjvec}\lexval_i(s) = \quanfun_i(s)$.
	Note the following corner case: If $s$ was in all target sets, then $\objvec(s)$ would be empty, which is not covered by the induction hypothesis. However, if $s$ is in all target sets, $^{\objvec(s)}\lexval_{i}(s)$ is never used in the definition, but $\quanfun$ is set to 1 everywhere.
	\item Now let $s \notin \alltargets$.
	Let $\bothstrats$ be any strategies of $\plmax$ and $\plmin$, respectively. Let $\obj_i = \reach{S_i}$. In the induced Markov chain $\game^\bothstrats$, the following holds (\textit{all infima and suprema are taken over lex-optimal (counter-)strategies with respect to $\objvec_{<i}$}):
	
	\begin{align*}
	\prob^\bothstrats_s(\obj_i) &= \sum_{\path t \in Paths_{fin}(\alltargets)} \prob_s^\bothstrats(\path t) \cdot \prob_{\path t}^{\bothstrats}(\obj_i) \aligncomment{because $\path$ visits no state in $\alltargets$}\\
	&= \sum_{\path t \in Paths_{fin}(\alltargets)} \prob_s^\bothstrats(\path t) \cdot \prob_{t}^{{\maxstrat(\path),\minstrat(\path)}}(\obj_i) \aligncomment{where $\maxstrat(\path)$ behaves like $\maxstrat$ after seeing $\path$}\\
	\implies \lexval_i(s)&= \sup_\maxstrat \inf_\minstrat \sum_{\path t \in Paths_{fin}(\alltargets)} \prob_s^\bothstrats(\path t) \cdot \sup_{\maxstrat(\path)} \inf_{\maxstrat(\path)} \prob_{t}^{{\maxstrat(\path),\minstrat(\path)}}(\obj_i) \aligncomment{because behavior of strategies after $\path$ is independent from that before $\path$; also recall that sup and inf is over lex-optimal strategies w.r.t. $\objvec_{<i}$} \\
	&= \sup_\maxstrat \inf_\minstrat \sum_{\path t \in Paths_{fin}(\alltargets)} \prob_s^\bothstrats(\path t) \cdot \ ^{\objvec}\lexval_i(t) \aligncomment{by definition of the lex-value $^{\objvec}\lexval(t)$} \\
	&= \sup_\maxstrat \inf_\minstrat \sum_{\path t \in Paths_{fin}(\alltargets)} \prob_s^\bothstrats(\path t) \cdot  \quanfun_i(t) \aligncomment{because $t \in \alltargets$ and by the argumentation above} \\
	&= \sup_\maxstrat \inf_\minstrat \sum_{t \in \alltargets} \prob_s^\bothstrats(\until{(S \setminus \alltargets)}{t} ) \cdot  \quanfun_i(t) \aligncomment{by definition of the until property} \\
	&=\sup_\maxstrat \inf_\minstrat \prob_s^\bothstrats(\reach{\quanfun_i})
	=\ ^{\qobjvec}\lexval_i(s).
	\end{align*}
	where we used the notation $Paths_{fin}(\alltargets) = \{\path t \in ((S \setminus \alltargets)\times \actlabels)^* \times S \mid  t \in \alltargets\}$ for the set of all finite paths to a state in $\alltargets$ in $\game^\bothstrats$ and $\prob_s^\bothstrats(\path t)$ denotes the probability of such a path when the Markov chain $\game^\bothstrats$ starts in $s$.
	
	If $\obj_i$ is a safety objective instead, then the argument is similar (recall that the semantics of a quantified safety objective is defined as $\prob_s^\bothstrats(\safe{\quanfun}) = 1 - \prob_s^\bothstrats(\reach{\quanfun})$).
	
	\begin{align*}
	\prob^\bothstrats_s(\safe{S_i}) &= 1 - \sum_{\path t \in Paths_{fin}(\alltargets)} \prob_s^\bothstrats(\path t) \cdot \prob_{\path t}^{\bothstrats}(\reach{S_i}) \\
	&= 1- \sum_{\path t \in Paths_{fin}(\alltargets)} \prob_s^\bothstrats(\path t) \cdot \prob_{t}^{{\maxstrat(\path),\minstrat(\path)}}(\reach{S_i})\\
	\implies \lexval_i(s)&= 1 - \inf_\maxstrat \sup_\minstrat \sum_{\path t \in Paths_{fin}(\alltargets)} \prob_s^\bothstrats(\path t) \cdot \inf_{\maxstrat(\path)} \sup_{\maxstrat(\path)} \prob_{t}^{{\maxstrat(\path),\minstrat(\path)}}(\reach{S_i}) \\
	&= 1 - \inf_\maxstrat \sup_\minstrat \sum_{\path t \in Paths_{fin}(\alltargets)} \prob_s^\bothstrats(\path t) \cdot \ (1 - \sup_{\maxstrat(\path)} \inf_{\maxstrat(\path)} \prob_{t}^{{\maxstrat(\path),\minstrat(\path)}}(\safe{S_i})) \\
	&= 1 - \inf_\maxstrat \sup_\minstrat \sum_{\path t \in Paths_{fin}(\alltargets)} \prob_s^\bothstrats(\path t) \cdot \ (1 - ^{\objvec}\lexval_i(t))  \\
	&=1 - \inf_\maxstrat \sup_\minstrat \sum_{\path t \in Paths_{fin}(\alltargets)} \prob_s^\bothstrats(\path t) \cdot  \quanfun_i(t) \\
	&= 1 - \inf_\maxstrat \sup_\minstrat \sum_{t \in \alltargets} \prob_s^\bothstrats(\until{(S \setminus \alltargets)}{t} ) \cdot  \quanfun_i(t) \\
	&=1 - \inf_\maxstrat \sup_\minstrat \prob_s^\bothstrats(\reach{\quanfun_i})\\
	&= \sup_\maxstrat \inf_\minstrat \prob_s^\bothstrats(\safe{\quanfun_i})
	=\ ^{\qobjvec}\lexval_i(s).
	\end{align*}
	
\end{itemize}
\qed}

\end{document}